\title{Equilibrium Behaviors in Repeated Games\thanks{We are grateful to Drew Fudenberg, Guillermo  Ordo\~{n}ez, Matthew Thomas, Linh To, and two anonymous referees for helpful comments. Pei acknowledges the National Science Foundation (Grant SES-1947021) for financial support.}}
\author{Yingkai Li\footnote{Department of Computer Science, Northwestern University. Email: yingkai.li@u.northwestern.edu} \and Harry Pei\footnote{Department of Economics, Northwestern University. Email: harrydp@northwestern.edu}}
\date{\today}

\documentclass[11pt]{article}
\usepackage{amsmath}
\usepackage{graphicx}
\usepackage{amsfonts}
\usepackage{amsthm}
\usepackage{setspace}
\usepackage{tikz}
\usepackage{amssymb}
\usetikzlibrary{patterns}
\usepackage{sgame}
\usepackage{color}
\usepackage{hyperref}
\usepackage{times}
\usepackage{enumitem}
\usepackage{mathptmx}
\usepackage{natbib}
\usepackage{tikz}
\usetikzlibrary{patterns}
\usetikzlibrary{shapes}
\usetikzlibrary{plotmarks}
\setcitestyle{authoryear,open={(},close={)}}

\usepackage[bottom]{footmisc}
\usepackage[top=1in, bottom=1in, left=0.9in, right=0.9in]{geometry}

\usepackage{ifthen}

\newcommand{\given}{\,\mid\,}

\newcommand{\prob}[2][]{\text{\bf Pr}\ifthenelse{\not\equal{}{#1}}{_{#1}}{}\!\left[{\def\givenn{\middle|}#2}\right]}
\newcommand{\expect}[2][]{\text{\bf E}\ifthenelse{\not\equal{}{#1}}{_{#1}}{}\!\left[{\def\givenn{\middle|}#2}\right]}

\newcommand{\tparen}{\big}
\newcommand{\tprob}[2][]{\text{\bf Pr}\ifthenelse{\not\equal{}{#1}}{_{#1}}{}\tparen[{\def\given{\tparen|}#2}\tparen]}
\newcommand{\texpect}[2][]{\text{\bf E}\ifthenelse{\not\equal{}{#1}}{_{#1}}{}\tparen[{\def\given{\tparen|}#2}\tparen]}

\newcommand{\sprob}[2][]{\text{\bf Pr}\ifthenelse{\not\equal{}{#1}}{_{#1}}{}[#2]}
\newcommand{\sexpect}[2][]{\text{\bf E}\ifthenelse{\not\equal{}{#1}}{_{#1}}{}[#2]}

\newcommand{\mixaction}{\alpha}

\newcommand{\commitu}{u_1(a^*,b^*)}
\newcommand{\strategy}{\sigma}
\newcommand{\minmaxu}{\underline{\nu}_1}
\newcommand{\event}{{\cal E}}
\newcommand{\deviatea}{a'}

\newcommand{\q}{p}
\newcommand{\abs}[1]{\left\lvert #1\right\rvert}

\begin{document}
\maketitle
\numberwithin{equation}{section}

\begin{abstract}
\normalsize
We examine a patient player's behavior when he can build  reputations in front of a sequence of myopic opponents. With positive probability, the patient player is a commitment type who plays his Stackelberg action in every period. We characterize the patient player's action frequencies in equilibrium. Our results clarify the extent to which reputations can refine the patient player's behavior and provide new insights to entry deterrence, business transactions, and capital taxation. Our proof makes a methodological contribution by establishing a new concentration inequality.\\


\end{abstract}

\begin{spacing}{1.5}
\newtheorem{Proposition}{\hskip\parindent\bf{Proposition}}
\newtheorem{Theorem}{\hskip\parindent\bf{Theorem}}
\newtheorem{Lemma}{\hskip\parindent\bf{Lemma}}[section]
\newtheorem{Corollary}{\hskip\parindent\bf{Corollary}}[section]
\newtheorem{Definition}{\hskip\parindent\bf{Definition}}
\newtheorem{Assumption}{\hskip\parindent\bf{Assumption}}
\newtheorem{Condition}{\hskip\parindent\bf{Condition}}
\newtheorem{Claim}{\hskip\parindent\bf{Claim}}

\section{Introduction}\label{sec1}
Economists have long recognized that individuals, firms, and governments can benefit from good reputations. As shown in the seminal work of \citet{FL-89}, a patient player can guarantee himself a high payoff when his opponents believe that he might be committed to play a particular action. Their result can be viewed as a refinement, which selects the patient player's optimal equilibria in many games of interest.


This paper studies the effects of reputations on the patient player's behavior instead of his payoffs, which have been underexplored in the reputation literature.
Existing works on reputation-building behaviors restrict attention to particular equilibria or
games with particular payoff functions. By contrast,  
we identify tight bounds on the patient player's action frequencies that apply to all equilibria under more
general payoff functions. Our results clarify the extent to which reputations can refine the patient player's behavior and provide new insights to applications such as 
entry deterrence, business transactions, and capital taxation.

We analyze a repeated game between a patient player and a sequence of myopic opponents. The patient player is either a strategic type who maximizes his discounted average payoff, or a commitment type who plays his optimal pure commitment action (or \textit{Stackelberg action}) in every period. The myopic players cannot observe the patient player's type, but can observe 
all the actions taken in the past.

We examine the extent to which the option to imitate the commitment type can motivate the patient player to play his Stackelberg action. Theorem \ref{Theorem1} characterizes tight bounds on the discounted frequencies with which the strategic-type patient player plays his Stackelberg action in equilibrium. We show that the maximal frequency equals one and the minimal frequency equals the value of the following linear program: Choose a distribution over action profiles in order to minimize the probability of the Stackelberg action subject to two constraints. First, each action profile in the support of this distribution satisfies the myopic player's incentive constraint. Second, the patient player's expected payoff from this distribution is no less than his Stackelberg payoff. 
The first constraint is necessary since the myopic players best reply to the patient player's action in every period. The second constraint is necessary since
the patient player can approximately attain his Stackelberg payoff by imitating the commitment type. In order to provide him an incentive not to play his Stackelberg action, his continuation value after separating from the commitment type must be at least his Stackelberg payoff.

The substantial part is to show that these constraints are not only necessary but also sufficient. 
Our proof is constructive and makes a methodological contribution by establishing a novel concentration inequality on the discounted sum of random variables that 
bounds the patient player's action frequencies (Lemma \ref{lem:concentration}).

Theorem \ref{Theorem2} identifies a sufficient condition under which a distribution of the patient player's actions is his action frequency in some equilibria of the reputation game.
In a number of leading applications such as the product choice game and the entry deterrence game, our sufficient condition is also necessary, in which case Theorem \ref{Theorem2} fully characterizes of the set of action frequencies that can arise in equilibrium.

Our results provide new insights to classic applications of reputation models. For example,
in the product choice game of Mailath and Samuelson (2006, Figure 15.1.1 on page 460),\footnote{In \cite{MS-06}'s product choice game, a patient firm faces a sequence of consumers. In every period, the firm chooses between high effort and low effort, and a consumer chooses between buying a high-end product and a low-end product. The firm finds it costly to exert high effort and prefers the consumers to purchase high-end products. Each consumer has an incentive to buy the high-end product only when she believes that the firm will exert high effort with high enough probability. In this game, high effort is the firm's Stackelberg action but low effort is the dominant action in the stage game.} 
our results imply that a policy maker can increase the frequency of high effort by subsidizing consumers for purchasing low-end products or by taxing consumers for purchasing high-end products. Intuitively, these policies increase the consumers' demand for high effort when they purchase the high-end product, which in turn increases the frequency of high effort in the worst equilibrium.
In the entry deterrence game of \cite{KW-82} and \cite{MR-82}, our results imply that a small amount of subsidy to potential entrants for entering the market makes a reputation-building incumbent more aggressive in fighting entry, but a large amount of subsidy eliminates the incumbent's fighting incentives.

Our results contribute to the reputation literature by clarifying the role of reputations in refining the patient player's behavior. This is complementary to the result of \cite{FL-89} that studies how reputations refine the patient player's payoff. 
Existing works on players' reputation-building behaviors restrict attention to particular equilibria or particular payoff functions. 
For example, \citet{KW-82} and \citet{MR-82}
characterize sequential equilibria
in entry deterrence games. 
\citet{Sch-93} characterizes Markov equilibria in repeated bargaining games. \cite{Bar-03},
\citet{Phe-06}, \citet{Ekm-11}, \citet{Liu-11}, and \citet{LS-14}
restrict attention to supermodular games or $2 \times 2$ games. By contrast,
we characterize tight bounds on the patient player's action frequencies that apply to all equilibria. Our results are more general in terms of payoffs, which only require the patient player's optimal commitment payoff to be greater than his minmax value and that his optimal commitment outcome is not a stage-game Nash equilibrium.


\citet*{CMS-04} show that when the monitoring structure has full support, the myopic players eventually learn the patient player's type and the strategies converge to an equilibrium of the repeated complete information game. However, their results do not characterize the speed of convergence or players' behaviors in finite time, and hence do not imply what players' discounted action frequencies are. 
\citet{EM-19} study players' reputation-building behaviors in stopping games where a patient uninformed player chooses  between continuing and irreversibly stopping the game in every period. By contrast, the uninformed players in our model are myopic and their action choices are reversible.
\citet{Pei2020} provides sufficient conditions under which the patient player has a unique on-path behavior. Unlike our model that restricts attention to private value environments but allows for general stage-game payoffs, his result requires nontrivial interdependent values and monotone-supermodular stage-game payoffs.

Section \ref{sec2} sets up the baseline model. Section \ref{sec3} states our main results. Section \ref{sec4} applies our results to several applied models of reputation formation and discusses the results' practical implications. Section \ref{sec6} discusses our modeling assumptions as well as issues related to taking our predictions to the data. Section \ref{sec7} concludes. The proofs of our results can be found in the appendix.  

\section{Model}\label{sec2}
Time is discrete, indexed by $t=0,1,2,...$. A patient player $1$ with discount factor $\delta \in (0,1)$ interacts with an infinite sequence of myopic player $2$s, arriving one in each period and each playing the game only once.
In period~$t$, a public randomization device $\xi_t \sim U[0,1]$ is realized and is observed by both players, after which players simultaneously choose their actions. Player $1$'s action is denoted by $a_t \in A$. Player $2$'s action is denoted by $b_t \in B$. Their stage-game payoffs are $u_1(a_t,b_t)$ and $u_2(a_t,b_t)$. We assume $A$ and $B$ are finite, with $|A|,|B|\geq 2$.

Let $\textrm{BR}_1: \Delta (B) \rightrightarrows 2^{A} \backslash \{\varnothing\}$ and $\textrm{BR}_2: \Delta (A) \rightrightarrows 2^{B} \backslash \{\varnothing\}$ be player $1$'s and player $2$'s best reply correspondences in the stage-game. The set of player $1$'s (pure) Stackelberg actions is
$\arg\max_{a \in A} \{ \min_{b \in \textrm{BR}_2(a)} u_1 (a,b) \}$.
\begin{Assumption}\label{Ass1}
Player $1$ has a unique Stackelberg action, denoted by $a^*$. 
Player $2$ has a unique best reply to player $1$'s Stackelberg action, denoted by $b^*$. 
\end{Assumption}
 Assumption \ref{Ass1} is satisfied 
when each player has a strict best reply to each of his opponent's pure actions and player $1$ is not indifferent between any pair of pure action profiles, both of which are satisfied for generic $(u_1,u_2)$ since $A$ and $B$ are finite sets. Player $1$'s \textit{Stackelberg payoff} is $u_1(a^*,b^*)$. 
Let
\begin{equation*}
    \mathcal{B} \equiv \{\beta \in \Delta(B) | \exists \alpha \in \Delta(A) \textrm{ s.t. } \textrm{supp}(\beta) \subset \textrm{BR}_2(\alpha)\} \subset \Delta (B).
\end{equation*}
Since player $2$s are myopic,
they will never take actions that do not belong to
$\mathcal{B}$. As a result, player $1$'s minmax value is $\underline{v}_1 \equiv  \min_{\beta \in  \mathcal{B}} \max_{a \in A} u_1(a, \beta)$.
\begin{Assumption}\label{Ass2}
$a^* \notin \textrm{BR}_1(b^*)$ and $u_1(a^*,b^*) > \underline{v}_1$.
\end{Assumption}
Assumptions \ref{Ass1} and \ref{Ass2} are satisfied in many leading applications of reputation models. For example, 
\begin{enumerate}
    \item In the product choice game of \citet{MS-06}, a firm benefits from committing to exert high effort since it can encourage consumers to purchase the high-end product or to purchase larger quantities.
    However, the firm can save costs by lowering its effort.
    \item In the entry deterrence game of \citet{KW-82} and \citet{MR-82}, and the limit pricing game of \cite{MR-82ECMA}, an incumbent firm  benefits from committing to set low prices and to fight potential entrants, but its stage-game payoff is higher when it accommodates entry. 
        \item In the fiscal policy game of \cite{Phe-06}, the government benefits from committing to low tax rates in order to encourage investments, but it is tempted to expropriate the citizens after investment takes place. 
    \item In the monetary policy game of \citet{Bar-86}, the central bank can benefit from committing to low inflation rates. But given the households' expectations about inflation, the central bank is tempted to raise inflation in order to boost economic activities.
\end{enumerate}

Assumption \ref{Ass2} rules out coordination games (such as the battle of sexes), common interest games, and chicken games, in which $a^*$ best replies to $b^*$, and zero-sum games in which $u_1(a^*,b^*)  \leq \underline{v}_1$.  
Section \ref{sec6} discusses games that violate this assumption, and the role of Assumption \ref{Ass2} in our proofs is explained in Appendix \ref{sec5}.

Player $1$ has perfectly persistent private information about his type $\omega$. Let $\omega \in \{\omega^s, \omega^c\}$, where $\omega^c$ stands for a \textit{commitment type} who mechanically plays $a^*$ in every period,
and $\omega^s$ stands for a \textit{strategic type} who  can flexibly choose his actions in order to maximize his discounted average payoff $\sum_{t=0}^{+\infty} (1-\delta)\delta^t u_1(a_t,b_t)$.
Player $2$'s prior belief attaches probability $\pi \in (0,1)$ to the commitment type.

Players' past actions are perfectly monitored.
A typical public history is denoted by $h^t \equiv \{a_s,b_s,\xi_s\}_{s=0}^{t-1}$. Let $\mathcal{H}^t$ be the set of $h^t$ and let
$\mathcal{H} \equiv \cup_{t \in \mathbb{N}} \mathcal{H}^t$. Strategic-type player $1$'s strategy is $\sigma_1: \mathcal{H} \rightarrow \Delta (A)$. Player $2$'s strategy is $\sigma_2 : \mathcal{H} \rightarrow \Delta (B)$.
Let $\Sigma_1$ and $\Sigma_2$ be the set of player $1$'s and player $2$'s strategies, respectively.

The solution concept is (Bayes) Nash equilibrium. Let $\textrm{NE}(\delta,\pi) \subset \Sigma_1 \times \Sigma_2$ be the set of equilibria. Since the stage game is finite and 
payoffs are discounted, an equilibrium exists \citep{fl83}. 

\paragraph{Existing Result on Equilibrium Payoffs:}  \citet{FL-89} show that for every $\pi \in (0,1)$ and $\varepsilon>0$, there exists $\underline{\delta} \in (0,1)$ such that
\begin{equation}\label{2.1}
    \inf_{(\sigma_1,\sigma_2) \in \textrm{NE}(\delta,\pi)} \mathbb{E}^{(\sigma_1,\sigma_2)}\Big[
    \sum_{t=0}^{+\infty} (1-\delta) \delta^t u_1(a_t,b_t)
    \Big] \geq u_1(a^*,b^*)- \varepsilon \textrm{ for every } \delta > \underline{\delta},
\end{equation}
where $\mathbb{E}^{(\sigma_1,\sigma_2)}[\cdot]$ is the
expectation when player $1$'s strategy is $\sigma_1$ 
and player $2$'s strategy is $\sigma_2$.

Inequality (\ref{2.1}) unveils the effects of reputations on the patient player's payoff. \cite{FL-89} view this result as a refinement, which selects among the plethora of equilibria in repeated complete information games.
According to the folk theorem of
\citet*{FKM-90}, the patient player can attain any payoff between $\underline{v}_1$ and
$\overline{v}_1 \equiv \max_{ \{ (\alpha,\beta) |  \textrm{supp}(\beta) \subset \textrm{BR}_2(\alpha) \}} \min_{a \in \textrm{supp}(\alpha)} u_1(a,\beta)$ in a repeated complete information game without any commitment type.
By definition, $\overline{v}_1 \geq u_1(a^*,b^*)$, which implies that introducing a commitment type selects equilibria in which player $1$'s payoff is between $u_1(a^*,b^*)$ and $\overline{v}_1$. In the entry deterrence game, product choice game, and fiscal and monetary policy games, 
$\overline{v}_1$ equals $u_1(a^*,b^*)$, in which case the reputation model selects equilibria where the patient player receives his highest equilibrium payoff.

\section{Results}\label{sec3}
Our results examine the \textit{discounted frequencies} of the patient player's actions. Formally,
the discounted frequency of  action $a \in A$ under  $(\sigma_1,\sigma_2)$ is
\begin{equation}\label{3.1}
G^{(\sigma_1,\sigma_2)}(a) \equiv  \mathbb{E}^{(\sigma_1,\sigma_2)} \Big[
    \sum_{t=0}^{\infty} (1-\delta)\delta^t \mathbf{1}\{a_{t}=a\}
    \Big].
\end{equation}
Our first result characterizes the discounted frequencies with which
the patient player plays his Stackelberg action $a^*$. Let
\begin{equation}\label{3.2}
    \Gamma \equiv \Big\{  (\alpha,b) \in \Delta (A) \times B \Big| b \in \textrm{BR}_2(\alpha) \Big\}
\end{equation}
be the set of \textit{incentive compatible action profiles}.
Let
\begin{equation}\label{3.3}
    F^* (u_1,u_2) \equiv \min_{(\alpha_1,\alpha_2,b_1,b_2,q) \in \Delta (A) \times \Delta (A) \times B \times B \times [0,1] }
  \Big\{   q \alpha_1(a^*) +(1-q) \alpha_2(a^*) \Big\},
\end{equation}
subject to
\begin{equation}\label{3.4}
(\alpha_1,b_1) \in \Gamma,\quad (\alpha_2,b_2) \in \Gamma,
\end{equation}
and
\begin{equation}\label{3.5}
    q u_1(\alpha_1,b_1) +(1-q) u_1(\alpha_2,b_2) \geq u_1(a^*,b^*),
\end{equation}
where $\alpha_i(a)$ stands for the probability
of action $a \in A$ in $\alpha_i \in \Delta (A)$.
\begin{Theorem}\label{Theorem1}
Suppose $(u_1,u_2)$ satisfies Assumptions 1 and 2. 
\begin{enumerate}
    \item For every $f \in [F^*(u_1,u_2),1]$ and $\varepsilon>0$, there exists $\underline{\delta} \in (0,1)$ such that for every $\delta  > \underline{\delta}$, there exists $(\sigma_1,\sigma_2) \in \textrm{NE}(\delta,\pi)$ such that $G^{(\sigma_1,\sigma_2)}(a^*) \in (f-\varepsilon,f+\varepsilon)$.
     \item For every $\widehat{f}< F^*(u_1,u_2)$, there exist $\underline{\delta} \in (0,1)$ and $\eta>0$ such that $G^{(\sigma_1,\sigma_2)}(a^*) > \widehat{f}+\eta$
    for every $\delta > \underline{\delta}$ and $(\sigma_1,\sigma_2) \in \textrm{NE}(\delta,\pi)$.
\end{enumerate}
\end{Theorem}
Theorem \ref{Theorem1} implies that when player $1$ is patient, the discounted frequency with which he plays $a^*$ can take any value between 
$F^*(u_1,u_2)$ and $1$, but it cannot be strictly lower than $F^*(u_1,u_2)$.
Therefore, $[F^*(u_1,u_2),1]$
is the set of frequencies with which $a^*$ can arise in equilibrium. 
Our result applies to every prior belief $\pi \in (0,1)$, which includes but not limited to situations where the probability of the commitment type is small.
Since  $F^*(u_1,u_2)<1$ under Assumption \ref{Ass2}, Theorem \ref{Theorem1} implies 
that an arbitrarily patient player can play his Stackelberg action with frequency bounded away from one despite having the option to build a reputation.

The upper bound on the frequency of $a^*$ is $1$
since there exists an equilibrium where player $1$ plays $a^*$ and player $2$s play $b^*$. Once player $1$ plays any action other than $a^*$, future player $2$s can observe this deviation after which they can punish player $1$ by driving his continuation value to his minmax payoff $\underline{v}_1$. Such a punishment is feasible since player $1$ separates from the commitment type after any deviation from his equilibrium strategy, and according to \citet*{FKM-90}, there exists an equilibrium of the repeated complete information game in which player $1$'s payoff is $\underline{v}_1$. Since Assumption \ref{Ass2} requires that $u_1(a^*,b^*)>\underline{v}_1$, this punishment provides player $1$ an incentive to play $a^*$ when his discount factor $\delta$ is large enough.

For some intuition on the linear program that defines 
the lower bound $F^*(u_1,u_2)$, consider a static planning problem in which a planner commits to a mixed action $\alpha \in \Delta (A)$ on behalf of player~$1$ after which player $2$ best replies to $\alpha$. Suppose the planner faces a constraint that player $1$'s expected payoff is no less than $u_1(a^*,b^*)$, then by definition, $F^*(u_1,u_2)$ is the lowest probability with which $a^*$ needs to be played.\footnote{The planner in the planning problem can randomize between any number of commitment actions, while in the linear program that defines $F^*(u_1,u_2)$, he can randomize between at most two commitment actions. Lemmas \ref{LA.2} and \ref{LA.3} show that this is without loss and the value of $F^*(u_1,u_2)$ remains the same even when the planner can randomize between any arbitrary number of commitment actions.}

We map the two constraints in the planning problem to the reputation game studied by Theorem \ref{Theorem1}. First, since player $2$s are myopic, they play a best reply to $\alpha$ after they learn that the patient player will play $\alpha$. This explains the necessity of constraint (\ref{3.4}). Second, the presence of commitment type implies that the patient player can guarantee payoff approximately $u_1(a^*,b^*)$ by playing $a^*$ in every period. Therefore, the patient player has an incentive to play $\alpha_1$ with probability $q$ and $\alpha_2$ with probability $1-q$ only when his expected payoff from doing so is at least 
$u_1(a^*,b^*)$. This explains the necessity of constraint (\ref{3.5}).
The substantial part of our result is to show that constraints (\ref{3.4}) and (\ref{3.5}) are not only necessary but are also sufficient.

Our second result examines the set of discounted action frequencies that can arise in equilibrium.
Let
\begin{equation}\label{eq:A}
    \mathcal{A} \equiv \Big\{
    \alpha^* \in \Delta(A)
    \Big|
    \exists q \in \Delta (\Gamma) \textrm{ such that }
\alpha^*    = \int_{\alpha} \alpha d q
\textrm{ and } \int_{(\alpha,b)} u_1(\alpha,b) d q = u_1(a^*,b^*)
    \Big\},
\end{equation}
which is the set of marginal distributions of player $1$'s actions such that
one can find a distribution of incentive compatible action profiles $q \in \Delta (\Gamma)$ from which  player $1$'s expected payoff equals his Stackelberg payoff.
\begin{Theorem}\label{Theorem2}
Suppose $(u_1,u_2)$ satisfies Assumptions \ref{Ass1} and \ref{Ass2}.
\begin{enumerate}
    \item For every $\alpha^* \in \mathcal{A}$ and $\varepsilon>0$, there exists $\underline{\delta} \in (0,1)$ such that for every $\delta > \underline{\delta}$, there exists $(\sigma_1,\sigma_2) \in \textrm{NE}(\delta,\pi)$ such that $\Big| G^{(\sigma_1,\sigma_2)}(a)- \alpha^* (a) \Big| < \varepsilon$ for every $a \in A$.\footnote{We can also show that if $\delta$ is large enough and $\mathcal{A}$ satisfies a full dimensionality assumption, then every $\alpha^*$ that belongs to the interior of $\mathcal{A}$ can be exactly attained as the discounted action frequency of some equilibria.}
\item In games where $u_1(a^*,b^*)=\overline{v}_1$. For every $\widehat{\alpha} \notin \mathcal{A}$, there exist $\eta>0$ and $\underline{\delta} \in (0,1)$ such that for every $\delta > \underline{\delta}$ and 
$(\sigma_1,\sigma_2) \in \textrm{NE}(\delta,\pi)$, $ \Big|  G^{(\sigma_1,\sigma_2)}(a)- \widehat{\alpha} (a) \Big| > \eta$ for some $a \in A$.
\end{enumerate}
\end{Theorem}
According to Theorem \ref{Theorem2}, every action distribution that belongs to $\mathcal{A}$ is arbitrarily close to the patient player's action frequency in some equilibria of the reputation game. In fact, the first statement of Theorem \ref{Theorem2} is a generalization of  Statement 1 of Theorem~\ref{Theorem1} since it is without loss of generality to focus on distributions of incentive compatible action profiles such that constraint (\ref{3.5}) is binding (Lemma \ref{LB.1}) and it is without loss of generality to focus on 
distributions supported on  $\Gamma$ that have at most two elements in their support when the objective is to minimize the discounted frequency of~$a^*$
(Lemma \ref{LA.2} and Lemma \ref{LA.3}).

In games where $u_1(a^*,b^*)=\overline{v}_1$, such as the product choice game and the entry deterrence game, an action distribution is the patient player's discounted action frequency in some equilibria \textit{if and only if} it belongs to $\mathcal{A}$. In this class of games, any action frequency that satisfies player $2$'s incentive constraints and yields player $1$ his Stackelberg payoff can be attained in some equilibria of the repeated game.

\section{Economic Applications}\label{sec4}
We apply our results to \textit{monotone-supermodular games} that include the leading applications of reputation models, such as the product choice game, the entry deterrence game, and the fiscal policy game.  
\begin{Definition}\label{Def1}
$(u_1,u_2)$ is monotone-supermodular if there exist a complete order on $A$ and 
a complete order on $B$ such that 
 $u_1(a,b)$ is strictly decreasing in $a$, and $u_2(a,b)$ has strictly increasing differences.\footnote{This definition resembles the one in \citet{LP-20} and \citet{Pei2020} except that there is no state that affects players' payoffs. We also do not require $u_1(a,b)$ to be strictly increasing in $b$.}
\end{Definition}
In order to facilitate the application of Theorem \ref{Theorem1},
we simplify the
linear program that defines $F^*(u_1,u_2)$.  Let $\underline{a}$ be the lowest element of $A$ and let $\underline{b} \in B$ be player $2$'s best reply to $\underline{a}$. 
If player $2$ has multiple best replies to $\underline{a}$, then let $\underline{b}$ the one that maximizes player $1$'s payoff.
Let
\begin{equation}\label{4.1}
    \Gamma^* \equiv \Big\{ (\alpha,b) \in \Gamma \Big|
|    \textrm{BR}_2(\alpha) | \geq 2 \textrm{ and }
b \in \arg\max_{b' \in \textrm{BR}_2(\alpha)} u_1(\alpha,b')
    \Big\}.
\end{equation} 
Intuitively, $\Gamma^*$ is a subset of $\Gamma$ that consists of  incentive compatible action profiles where player $2$ has at least two best replies, and for every $\alpha$ that player $2$ has multiple best replies, $b$ is the one that maximizes player $1$'s payoff. 
Under generic stage-game payoff functions, $\Gamma^*$ is a finite set. Proposition \ref{Prop1} implies that in games with monotone-supermodular payoffs, it is without loss of generality to choose incentive compatible action profiles from the finite set $\Gamma^* \cup \{\underline{a},\underline{b}\}$ instead of the infinite set $\Gamma$. 
\begin{Proposition}\label{Prop1}
If $(u_1,u_2)$ is monotone-supermodular, then
\begin{equation*}
    F^* (u_1,u_2) = \min_{(\alpha_1,\alpha_2,b_1,b_2,q) \in \Delta (A) \times \Delta (A) \times B \times B \times [0,1] }
  \Big\{   q \alpha_1(a^*) +(1-q) \alpha_2(a^*) \Big\},
\end{equation*}
subject to $(\alpha_1,b_1),(\alpha_2,b_2) \in \Gamma^* \cup \{(\underline{a},\underline{b})\}$, and $q u_1(\alpha_1,b_1) +(1-q) u_1(\alpha_2,b_2) \geq u_1(a^*,b^*)$. 
\end{Proposition}
The proof is in Appendix \ref{secC}. 
For the rest of this section, we apply our theorems as well as Proposition \ref{Prop1} to study product choice games, entry deterrence games, and capital taxation games.

\paragraph{Product Choice Game:} Player $1$ is a firm that chooses between high (action $H$) and low effort (action $L$). Player $2$s are consumers, each chooses between purchasing a
high-end product (action $h$) and  a low-end product (action $l$). Players' payoffs are:
\begin{center}
\begin{tabular}{| c | c | c |}
  \hline
  -- & $h$ & $l$ \\
  \hline
  $H$ & $1-c_h,2-\gamma^*$ & $-c_l,1$ \\
  \hline
  $L$ & $1, -\gamma^*$ & $0,0$ \\
  \hline
\end{tabular}
\end{center}
where 
$c_h,c_l \in (0,1)$ are the  costs of effort when the consumer buys the high-end product and the low-end product, respectively, and consumers are willing to choose $h$ only when they believe that the firm exerts high effort with probability more than $\gamma^* \in (0,1)$.

This game has monotone-supermodular payoffs
once we rank the firm's actions according to $H \succ L$ and the consumers' actions according to $h \succ l$.
The firm's Stackelberg action is $H$. According to (\ref{4.1}), $\Gamma^*$ is a singleton  set $\Big\{ (\gamma^*H +(1-\gamma^*)L , h) \Big\}$. Proposition \ref{Prop1} implies that
\begin{equation}\label{4.2}
    F^*(u_1,u_2) = \min_{q \in [0,1]} q \gamma^*, \quad \textrm{subject to} \quad q \gamma^* u_1(H, h) 
    +q (1-\gamma^*) u_1(L,h)
    + (1-q) u_1(L,l) \geq u_1(H,h),
\end{equation}
from which we obtain
\begin{equation}\label{4.3}
F^*(u_1,u_2)=\frac{\gamma^* (1-c_h)}{1-\gamma^* c_h}. 
\end{equation}
\begin{Claim}\label{C1}
The lowest discounted frequency with which the firm exerts high effort strictly increases in $\gamma^*$, strictly decreases in $c_h$, and is independent of $c_l$.
\end{Claim}
In terms of practical implications, consider a policy maker who wants to increase the frequency with which the firm exerts high effort but does not know which equilibrium players coordinate on. The policy maker is ambiguity averse and evaluates the effectiveness of each policy according to the frequency of high effort in the \textit{worst equilibrium}. That is, his objective is to increase $F^*(u_1,u_2)$.

Claim \ref{C1} implies that the policy maker can increase $F^*(u_1,u_2)$ by subsidizing consumers for purchasing the low-end product or by taxing consumers for purchasing the high-end product.
Intuitively, these policies increase the consumers' demand for high effort when they purchase the high-end product. This leads to an increase in the equilibrium frequency of high effort since the firm needs to induce consumers to purchase the high-end product with high enough probability in order to obtain its Stackelberg payoff.

Next, we consider a variant of the product choice game in which every consumer chooses whether to buy a high-end product, an intermediate product, or a low-end product. The firm's payoffs are:
\begin{center}
\begin{tabular}{| c | c | c | c |}
  \hline
  -- & $h$ & $m$ & $l$\\
  \hline
  $H$ & $1-c$ & $p-c$ & $-c$\\
  \hline
  $L$ & $1$ & $p$ & $0$\\
  \hline
\end{tabular}
\end{center}
where its cost of effort is $c \in (0,1)$, its benefit from selling the high-end product is $1$, its benefit from selling the intermediate product is $p \in (0,1)$, and its benefit from selling the low-end product is  $0$. 

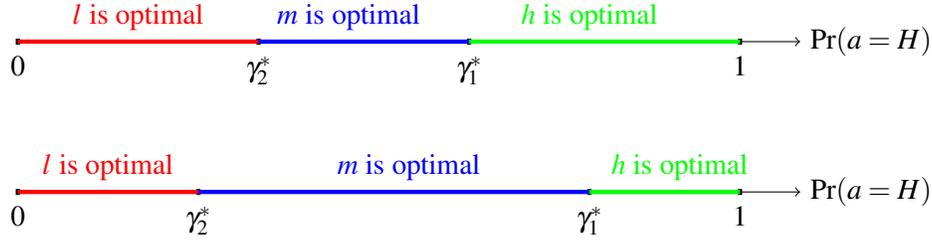
\begin{figure}
\begin{center}
\begin{tikzpicture}[scale=0.4]
\draw[->] (0,0)--(26,0)node[right]{$\Pr(a=H)$};
\draw[ultra thick] (24,0.1)--(24,-0.1)node[below]{$1$};
\draw[ultra thick] (0,0.1)--(0,-0.1)node[below]{$0$};
\draw[ultra thick] (8,0.1)--(8,-0.1)node[below]{$\gamma_2^*$};
\draw[ultra thick] (15,0.1)--(15,-0.1)node[below]{$\gamma_1^*$};
\draw[ultra thick, red] (0,0)--(4,0)node[above]{$l$ is optimal}--(8,0);
\draw[ultra thick, blue] (8,0)--(11,0)node[above]{$m$ is optimal}--(15,0);
\draw[ultra thick, green] (15,0)--(19,0)node[above]{$h$ is optimal}--(24,0);
\draw[->] (0,-5)--(26,-5)node[right]{$\Pr(a=H)$};
\draw[ultra thick] (24,-4.9)--(24,-5.1)node[below]{$1$};
\draw[ultra thick] (0,-4.9)--(0,-5.1)node[below]{$0$};
\draw[ultra thick] (6,-4.9)--(6,-5.1)node[below]{$\gamma_2^*$};
\draw[ultra thick] (19,-4.9)--(19,-5.1)node[below]{$\gamma_1^*$};
\draw[ultra thick, red] (0,-5)--(3,-5)node[above]{$l$ is optimal}--(6,-5);
\draw[ultra thick, blue] (6,-5)--(13,-5)node[above]{$m$ is optimal}--(19,-5);
\draw[ultra thick, green] (19,-5)--(22,-5)node[above]{$h$ is optimal}--(24,-5);
\end{tikzpicture}
\caption{Product choice game with three options: Consumer's best response before  (upper panel) and after they receive a subsidy for purchasing the intermediate product (lower panel).}
\end{center}
\end{figure}

The value of $F^*(u_1,u_2)$ depends on consumers' payoffs only through two sufficient statistics $\gamma_1^*$ and $\gamma_2^*$ with $0< \gamma_2^* < \gamma_1^* <1$, such that  a consumer has an incentive to choose $h$ when the firm exerts high effort with probability more than $\gamma_1^*$, has an incentive to choose $m$ when the firm exerts high effort with probability between $\gamma_2^*$ and $\gamma_1^*$, and has an incentive to choose $l$ when the firm exerts high effort with probability less than $\gamma_2^*$.

This game has monotone-supermodular payoffs once the firm's actions are ranked according to $H \succ L$ and consumers' actions are ranked according to $h \succ m \succ l$.
Applying Proposition \ref{Prop1} to this game, we have:
\begin{equation}\label{4.4}
F^*(u_1,u_2) = \begin{cases}
\frac{\gamma_1^* (1-c)}{1-\gamma_1^* c} & \textrm{ if } p \leq \frac{\gamma_2^*}{\gamma_1^*}\\
\frac{\gamma_2^* (1-c)}{p-\gamma_2^* c} & \textrm{ if } p > \frac{\gamma_2^*}{\gamma_1^*} \textrm{ and } c \geq \frac{1-p}{1-\gamma_2^*}\\
\frac{\gamma_1^* (1-p) -c(\gamma_1^*-\gamma_2^*)}{(1-p)-c(\gamma_1^*-\gamma_2^*)} & \textrm{ if }
 p > \frac{\gamma_2^*}{\gamma_1^*} \textrm{ and } c < \frac{1-p}{1-\gamma_2^*}.
\end{cases}
\end{equation}
Similar to the game with two purchasing options, we examine the effects of 
a small amount of sales taxes and subsidies for each product on  $F^*(u_1,u_2)$.
\begin{enumerate}
    \item A  tax on consumers for purchasing the high-end product (i.e., an increase in $\gamma_1^*$) has no effect on $F^*(u_1,u_2)$ when  $p > \frac{\gamma_2^*}{\gamma_1^*}$ and $c \geq \frac{1-p}{1-\gamma_2^*}$, and increases $F^*(u_1,u_2)$ otherwise. A subsidy on consumers for purchasing the low-end product (i.e., an increase in $\gamma_2^*$) has no effect on $F^*(u_1,u_2)$ when $ p \leq \frac{\gamma_2^*}{\gamma_1^*}$, and increases $F^*(u_1,u_2)$ otherwise.
    \item A subsidy on consumers for purchasing the intermediate product (i.e., a decrease in $\gamma_2^*$ and an increase in $\gamma_1^*$) leads to an increase in $F^*(u_1,u_2)$ when $ p \leq \frac{\gamma_2^*}{\gamma_1^*}$, leads to a decrease in $F^*(u_1,u_2)$ when $p > \frac{\gamma_2^*}{\gamma_1^*}$ and $c \geq \frac{1-p}{1-\gamma_2^*}$, and has an ambiguous effect on $F^*(u_1,u_2)$ when $p > \frac{\gamma_2^*}{\gamma_1^*}$ and $c < \frac{1-p}{1-\gamma_2^*}$.
\end{enumerate}

We obtain two additional insights compared to the case with two products. 
First, the effectiveness of  subsidizing low-end products depends on the firm's benefit from selling intermediate products (i.e., the comparison between $p$ and $\frac{\gamma_2^*}{\gamma_1^*}$).
This is because
a small subsidy for purchasing the low-end product only increases the demand for effort when the consumer decides whether to purchase the intermediate product instead of the low-end product, but does not affect consumers' demand for effort when deciding whether to purchase the high-end product instead of the intermediate product. 
When selling the intermediate product is unprofitable (i.e., $p \leq \frac{\gamma_2^*}{\gamma_1^*}$), an increase in the demand for effort when consumers decide between $l$ and $m$ does not affect the firm's equilibrium action frequencies. Similarly,
the effectiveness of taxing high-end products also depends on the profitability of selling the intermediate product, that is, the comparison between $c$ and $\frac{1-p}{1-\gamma_2^*}$.

Second, subsidizing consumers for purchasing intermediate products encourages the firm to exert effort more frequently when the firm's profit from selling the intermediate product is low (i.e., $p \leq \frac{\gamma_2^*}{\gamma_1^*}$), but encourages the firm to shirk more frequently otherwise. Intuitively, 
subsidizing the intermediate product has an effect similar to that of subsidizing the high-end product in the two-product setting when selling the intermediate product is attractive for the firm, and has an effect similar to that of subsidizing the low-end product when selling the intermediate product is unattractive.

\paragraph{Entry Deterrence Game:}  Player $1$ is an incumbent firm that chooses between fight (action $F$) and accommodate (action $A$). Player $2$s are potential entrants. Each of them chooses between staying out (action $O$) and entering the market (action $I$). Players' payoffs are:
\begin{center}
\begin{tabular}{| c | c | c |}
  \hline
  -- & $O$ & $I$ \\
  \hline
  $F$ & $1-c_o,0$ & $-c_i,-(1-\gamma^*)$ \\
  \hline
  $A$ & $1, 0$ & $0,\gamma^*$ \\
  \hline
\end{tabular}
\end{center}
where $c_o \in (0,1)$ is the incumbent's cost of setting low prices when the potential entrant stays out, and $c_i>0$ is its cost of setting low prices when the potential entrant enters.
Each potential entrant prefers to stay out only when the incumbent fights with probability more than $\gamma^* \in (0, 1)$.

These payoffs are monotone-supermodular
once we rank the incumbent's actions according to $F \succ A$, and the entrant's actions according to $O \succ I$.
The incumbent's Stackelberg action is $F$. Proposition \ref{Prop1} implies that:
\begin{equation*}
    F^*(u_1,u_2)= \frac{(1-c_o) \gamma^*}{ 1-c_o \gamma^*}.
\end{equation*}
\begin{Claim}\label{C2}
The lowest discounted frequency with which the incumbent fights potential entrants strictly increases in $\gamma^*$, strictly decreases in $c_o$, and 
is independent of 
$c_i$.
\end{Claim}
In terms of practical implications, 
consider a policy maker who can subsidize potential entrants for entering the market.
This is modeled as an increase in every entrant's payoff from action $I$ by $s>0$. 
Claim \ref{C2} implies that the frequency with which the incumbent fights entry is non-monotone with respect to the amount of subsidy. In particular,
\begin{enumerate}
    \item When the subsidy to potential entrants is close to but strictly less than $1-\gamma^*$, the strategic-type incumbent fights with frequency close to $1$ in \textit{all} equilibria. More generally, our formula implies that
    when $s < 1-\gamma^*$, a marginal increase in the amount of subsidy increases $F^*(u_1,u_2)$.
    \item When the subsidy is more than $1-\gamma^*$, 
each entrant has a strict incentive to enter the market regardless of the incumbent's action, so the incumbent
plays $A$ in every period. Therefore, the frequency with which the incumbent fights is zero in \textit{all} equilibria.
\end{enumerate}

\paragraph{Fiscal Policy Game:} Player $1$ is a government that chooses between a normal tax rate and a high tax rate (i.e., expropriation) and player $2$s are citizens who decide whether to invest. Players' payoffs are:
\begin{center}
\begin{tabular}{| c | c | c |}
  \hline
  -- & Invest & Not Invest \\
  \hline
  Normal Tax Rate & $\tau,1-\tau-c$ & $0,0$ \\
  \hline
  Expropriate & $1,-c$ & $0,0$ \\
  \hline
\end{tabular}
\end{center}
where the low tax rate is $\tau \in (0,1)$ and the cost of investment is $c \in (0,1-\tau)$. These payoffs are monotone-supermodular. The government's Stackelberg action is ``\textit{normal tax rate}'' and its Stackelberg payoff is $\tau$. 
According to Proposition \ref{Prop1}, the highest frequency
with which the government expropriates is:
\begin{equation*}
    1-F^*(u_1,u_2)=1-\frac{\tau}{1-\tau} \cdot \frac{c}{1-c},
\end{equation*}
which is a decreasing function of both $\tau$ and $c$.
This conclusion implies that in the worst case scenario, the frequency of government expropriation is lower when the government's revenue is higher under a normal tax rate (i.e., $\tau$ is larger), or when it is more costly for the citizens to invest (i.e., $c$ is larger).

\section{Discussions of Modeling Assumptions and Results}\label{sec6}
\paragraph{The Role of Assumption 2:} Assumption \ref{Ass2} rules out games in which the optimal commitment outcome $(a^*,b^*)$ is a stage-game Nash equilibrium (such as coordination games and chicken games), as well as games where player $1$'s optimal commitment payoff is no more than his minmax payoff (such as matching pennies). 

Our formula for the lowest discounted frequency of the Stackelberg action fails when $u_1(a^*,b^*) \leq \underline{v}_1$. For example, consider the following variant of the matching penny game that satisfies Assumption \ref{Ass1} and the first part of Assumption \ref{Ass2} but violates the second part of Assumption \ref{Ass2}:
\begin{center}
\begin{tabular}{| c | c | c |}
  \hline
  -- & $h$ & $t$ \\
  \hline
  $H$ & $1+\varepsilon,-1$ & $-1+\varepsilon,1$ \\
  \hline
  $T$ & $-1,1$ & $1,-1$ \\
  \hline
\end{tabular}
\end{center}
where $\varepsilon>0$. Player $1$'s unique Stackelberg action is $H$, his Stackelberg payoff is $-1+\varepsilon$, and his minmax payoff is close to $0$ when $\varepsilon$ is small enough. 
Therefore, $F^*(u_1,u_2)$ is close to $0$ when $\varepsilon$ is close to $0$.

However, if
both $\pi$ and $\varepsilon$ are small, then
the discounted frequency of action $H$ is close to $1/2$ in every equilibrium. This means that neither the lower bound 
 $F^*(u_1,u_2)$ nor the upper bound $1$  can be approximately attained in any equilibrium of the reputation game.

In games where $u_1(a^*,b^*)> \underline{v}_1$, but $(a^*,b^*)$ is a stage-game Nash equilibrium, our formula for the lowest discounted frequency for $a^*$ applies to the battle of sexes game and the chicken game, 
\begin{center}
\begin{tabular}{| c | c | c |}
  \hline
  Battle of Sexes & $o$ & $f$ \\
  \hline
  $O$ & $2,1$ & $0,0$ \\
  \hline
  $F$ & $0,0$ & $1,2$ \\
  \hline
\end{tabular}
\quad
\begin{tabular}{| c | c | c |}
  \hline
  Chicken Game & $h$ & $d$ \\
  \hline
  $H$ & $0,0$ & $7,2$ \\
  \hline
  $D$ & $2,7$ & $6,6$ \\
  \hline
\end{tabular}
\end{center}
or more generally, when $u_1(a^*,b^*)$ is player $1$'s highest feasible payoff and $u_1(a^*,b^*)>u_1(a,b)$ for every $(a,b) \neq (a^*,b^*)$. In those games, $F^*(u_1,u_2)=1$. This is because player $1$'s payoff is close to $u_1(a^*,b^*)$ in every equilibrium of the reputation game, so $a^*$ must be played with discounted frequency close to $1$. 

Next, we present a counterexample that satisfies Assumption \ref{Ass1} and the second part of Assumption \ref{Ass2} but violates the first part of Assumption \ref{Ass2}. 
Suppose players' payoffs are:
 \begin{center}
\begin{tabular}{| c | c | c |}
  \hline
  -- & $T$ & $N$ \\
  \hline
  $H$ & $1,1$ & $0,0$ \\
  \hline
  $M$ & $0,3$ & $3,0$ \\
  \hline
$L$ & $0,0$ & $0,3$ \\
  \hline
\end{tabular}
\end{center}
Player $1$'s Stackelberg action is $H$. Since $N$ is player $2$'s best reply to player $1$'s mixed action $\frac{1}{2}M + \frac{1}{2}L$, from which player $1$'s expected payoff is $3/2$, the value of $F^*(u_1,u_2)$ is $0$.

When the prior probability of commitment type $\pi$ is strictly greater than $3/4$, the discounted frequency with which player $1$ plays $H$ is $1$ in every equilibrium of the reputation game. This is because in every period where player $2$ has not observed player $1$ playing actions other than $H$, she has a strict incentive to play $T$, so player $1$'s payoff is $1$ by playing $H$ in every period. When player $1$ deviates to $M$ or $L$, his stage-game payoff is $0$, and his continuation value is no more than $1$ according to the folk theorem result of \citet*{FKM-90}. This implies that player $1$ plays $H$ at every on-path history in every equilibrium.

\paragraph{Mixed-Strategy Commitment Types:} Our model excludes commitment types that play mixed strategies. In order to understand the new challenges brought by mixed-strategy commitment types, consider the product choice game in Section \ref{sec4} where with positive probability, player $1$ is a type who mechanically plays $(\gamma^*+\varepsilon) H +(1-\gamma^*-\varepsilon) L$ in every period, where $\varepsilon>0$ is small.

A new complication arises since
the strategic type can never be separated from the mixed-strategy commitment type. As a result, the continuation game always has nontrivial incomplete information regardless of the strategies being played. This stands in contrast to games where all commitment types play pure strategies, in which the strategic type is separated from a commitment type as soon as he stops imitating that type.

Analyzing repeated games with persistent private information and short-lived uninformed players is a well-known challenge in the repeated games literature, and to the best of our knowledge, there is no existing result that characterizes the informed player's equilibrium behaviors or his equilibrium action frequencies.\footnote{Very few results are obtained in repeated games between an informed patient player and a sequence of uninformed myopic players. \cite{pei} characterizes the set of equilibrium payoffs between an informed seller and a sequence of uninformed buyers when the seller has persistent private information about his cost. His result relies on the assumption that all types of the seller have the same ordinal preference over stage-game outcomes, and does not apply when there are mixed-strategy commitment types.}

\paragraph{Rich Set of Commitment Types:} Our baseline model focuses on settings where there is only one commitment type. Our theorems extend to environments with any finite number of commitment types, as long as all of them play pure strategies,  and there exists a commitment type who  plays $a^*$ in every period.  

Our proof for the discounted frequency of action $a^*$ being no less than 
$F^*(u_1,u_2)$ remains the same. On the construction of equilibria that approximately attain a given frequency in $\mathcal{A}$, for every type space that satisfies the above requirements, 
there exists $T \in \mathbb{N}$ such that for every $\delta \in (0,1)$ and in every equilibrium under $\delta$, player $2$'s posterior belief in period $T$ assigns positive probability to at most one commitment type.
Construct the continuation equilibrium starting from period $T$ according to our proof in Appendix \ref{sec5}, 
the discounted frequency of player $1$'s action is close to $\alpha^* \in \mathcal{A}$ when
$\delta$ is close to $1$.

\paragraph{Testable Predictions:} Generally speaking, there are three challenges to test the predictions of reputation models.\footnote{Despite the large literature that takes repeated game predictions to the lab, see \cite{DF2018}, we are unaware of experimental results on repeated games with incomplete information between a patient player and a sequence of myopic players.}  First, econometricians do not know which equilibrium players coordinate on. 
Second, econometricians usually observe players' behaviors rather than their payoffs, while most of the existing reputation results that apply to all equilibria (such as those in Fudenberg and Levine 1989) are stated in terms of the patient player's payoff but not his behaviors.
Third, many interesting equilibria in reputation games are in mixed strategies, but econometricians usually cannot observe these mixed strategies and can only observe the realized pure strategy.

Our results overcome the first and the second challenge by delivering predictions on the patient player's action frequencies that apply to all equilibria. Take the product choice game example in Section \ref{sec4}. The expression for $F^*(u_1,u_2)$ depends only on two terms:
\begin{enumerate}
    \item $\gamma^*$: the minimal probability of high effort above which player $2$ is willing to play $h$;
    \item $c_h$: the ratio between the cost of effort
    and the firm's benefit when a consumer buys the high-end product.
\end{enumerate}
The values of $\gamma^*$ and $c_h$ can be computed without knowing all the details of players' stage-game payoff functions. Therefore, testing our predictions on the patient player's action frequencies has less demanding data requirements compared to testing the predictions on payoffs in canonical reputation models. 

In context of the product choice game between a firm and a sequence of consumers, one way to address the third challenge is to use the distribution of the firm's actions across different markets as a proxy for its mixed actions. This idea is applicable when the firm is a chain store that operates in many independent and geographically separated markets, and moreover, the consumers in each market can only observe the firm's actions in their own market but cannot observe the firm's actions in other markets. This is usually the case in developing countries where there is a lack-of record-keeping institutions, so that most consumers rely on word-of-mouth communication to learn about the firm's past behaviors. In these situations, it is reasonable to assume that consumers in one market cannot observe the firm's past behaviors in other markets. Using this idea, suppose an econometrician can observe the firm's behavior in every period and in every market, then he can compute the frequency of the firm's behaviors using his observations. He can then apply Theorems \ref{Theorem1} and \ref{Theorem2} to examine whether his observations are consistent with the predictions of reputation models.

The above discussion also unveils a limitation of our results, that they only characterize the set of action frequencies 
that can arise in equilibrium, but do not deliver predictions on the action frequencies that apply to \textit{every path of equilibrium play}. Therefore, an econometrician cannot test our predictions after observing a realized path of equilibrium play. He can do that after observing the firm's mixed actions, e.g.,  
observing the firm's behaviors across many markets and use the empirical distribution as a proxy for the firm's mixed action.

\section{Conclusion}\label{sec7}
We examine the effects of reputation on the frequencies with which a patient player plays each of his actions. Our results characterize tight bounds that apply to all equilibria in a broad class of games. 
Our research question stands in contrast to the 
reputation literature that focuses on the patient player's equilibrium payoff. Our results stand in contrast to those that study the patient player's behavior in some particular equilibria.

Our results imply that  in games where the optimal commitment outcome is not a stage-game Nash equilibrium, the patient player may play his optimal commitment action with frequency bounded away from one no matter how patient he is. When the patient player's optimal commitment payoff 
coincides with his highest equilibrium payoff in the repeated complete information game, reputation effects cannot further refine the patient player's behavior beyond that fact that his equilibrium payoff is at least his optimal commitment payoff. 

In terms of applications, our results imply that a policy maker can increase the frequency with which a firm exerts high effort by subsidizing consumers for purchasing low-end products or by taxing consumers for purchasing high-end products. They also imply that a small amount of subsidy to potential entrants for entering the market makes an incumbent more aggressive in fighting entrants, but a large amount of subsidy encourages the incumbent to accommodate entry. 

\newpage
\appendix
\section{Overview of Proofs}\label{sec5}
Our proof consists of two parts. Part 1 constructs a class of equilibria in which player $1$'s discounted action frequency is close to $\alpha^* \in \mathcal{A}$  when $\delta$ is close  to $1$.
Part 2 shows that $G^{(\sigma_1,\sigma_2)}(a^*)$ cannot be strictly lower than $F^*(u_1,u_2)$ in any equilibrium when $\delta$ is large enough, and in games where
$\overline{v}_1=u_1(a^*,b^*)$, any action distribution that does not belong to $\mathcal{A}$ cannot be player $1$'s action frequency in any equilibrium.

The first part of our proof makes a methodological contribution, where we establish a discounted version of the Wald's inequality
to bound the discounted frequency of each action.
We provide an overview of our equilibrium construction and explain our methodological contribution in this section, with details relegated to Appendix \ref{secB}.
The second part of our proof is standard, which we relegate to Appendix \ref{secA}.

\paragraph{Equilibrium Construction:}
The first part of Assumption \ref{Ass2} implies the existence of $a' \neq a^*$ such that $u_1(a',b^*) > u_1(a^*,b^*)$. Since $a^*$ is player $1$'s unique Stackelberg action, there exists $b' \neq b^*$ that best replies to $a'$ such that $u_1(a',b')< u_1(a^*,b^*)$.
Let $\alpha' \in \Delta \{a^*,a'\}$ be such that $\{b^*\}=\textrm{BR}_2(\alpha')$ and $u_1(\alpha',b^*)> u_1(a^*,b^*)$.

We construct a three-phase equilibrium in which the discounted frequency of player $1$'s actions is close to $\alpha^* \in \mathcal{A}$. Let $q \in \Delta (\Gamma)$ be a distribution of incentive compatible action profiles such that $\alpha^*=\int_{\alpha} \alpha dq$ and $\int_{(\alpha,b) \in \Gamma} u_1(\alpha,b) dq =u_1(a^*,b^*)$.
The equilibrium play starts from a \textit{preparation phase}, gradually reaches a \textit{normal phase}, and reaches a \textit{punishment phase} if and only if player $1$ has made an off-path deviation.
\begin{enumerate}
    \item Play belongs to the preparation phase when $t=0$, or when $t \geq 1$ and $(a^*,b^*)$ was played from period $0$ to $t-1$. In this phase, the strategic-type player $1$ plays $\alpha'$ and player $2$ plays $b^*$.
    \item Play belongs to the normal phase when there exists $s \leq t-1$ such that $(a_s,b_s) \neq (a^*,b^*)$. The normal phase consists of a number of \textit{blocks}, and players' strategies in each block will be specified later on.
    \item Player $1$'s continuation value when play first reaches the punishment phase is $\underline{v}_1$. This is feasible since player $2$'s belief attaches zero probability to the commitment type at every off-path history.
\end{enumerate}
In every block of the normal phase, $(\alpha',b^*)$ is played for the first $T_1 \in \mathbb{N}$ periods,
where $T_1$ is a constant that is independent of $\delta$. A \textit{review} happens by the end of these $T_1$ periods:
\begin{enumerate}
    \item If $(a',b^*)$ was not played in all $T_1$ periods, then play enters a \textit{compensation subphase}, where $(a',b')$ is played until period $T \in \mathbb{N}$ such that $(1-\delta) \sum_{t=0}^T \delta^t u_1(a_t,b_t) = (1-\delta^{T+1}) u_1(a^*,b^*)$.  The current block ends in period $T$ and the next block starts in period $T+1$.
    If there is no such integer $T$, then use the public randomization device in the last period that satisfies $(1-\delta) \sum_{t=0}^T \delta^t u_1(a_t,b_t) > (1-\delta^{T+1}) u_1(a^*,b^*)$.
    \item If $(a',b^*)$ was played in all  $T_1$ periods, then play enters an \textit{absorbing subphase}, in which   $(\alpha',b)$ is played with probability $\varepsilon_1>0$ and  $q \in \Delta (\Gamma)$ is played with complementary probability,  dictated by the realization of public randomization in the beginning of each period.
    The absorbing subphase ends in period $T$ where $T$ is the smallest integer that satisfies either
    \begin{equation*}
        \sum_{t=0}^T \delta^t u_1(a_t,b_t) < (1-\delta^{T+1}) u_1(a^*,b^*) + c(1-\delta),
    \end{equation*}
    or
    \begin{equation*}
      \sum_{t=0}^T \delta^t u_1(a_t,b_t) > (1-\delta^{T+1}) \Big(\varepsilon_1 u_1(\alpha',b^*) +(1-\varepsilon_1) \mathbb{E}_{(\alpha,b) \sim q} [u_1(\alpha,b)]\Big)-c(1-\delta),
    \end{equation*}
    where $c >0$ is a constant that is independent of $\delta$.
    Once the absorbing subphase ends, play enters the compensation subphase described in the first bulletin point, and the current block ends when $(1-\delta) \sum_{t=0}^T \delta^t u_1(a_t,b_t) = (1-\delta^{T+1}) u_1(a^*,b^*)$.
\end{enumerate}

One can verify that player $2$s' incentive constraints are satisfied.
The strategic type's  discounted average payoff is $u_1(a^*,b^*)$ from each of his on-path strategies, and  his continuation value at every on-path history is bounded away from $\underline{v}_1$.
The second part of Assumption \ref{Ass2} requires that $u_1(a^*,b^*)> \underline{v}_1$, which implies that player $1$ has no incentive to make any off-path deviations when $\delta$ is large enough. The two together verify player $1$'s incentive constraints.

The challenging step is to compute player $1$'s discounted action frequencies when he uses this history-dependent mixed strategy. To the best of our knowledge, the existing centrality results in probability theory either cannot handle geometric discounting (such as the Chernoff-Hoeffding's inequality) or do not provide tight bounds on the probability of concentration (such as the Lindeberg-Feller central limit theorem), making those  results inapplicable in our context.
We establish a novel concentration inequality that can overcome both challenges, which is also applicable to future studies of players' behaviors in dynamic games.
\begin{Lemma}\label{lem:concentration}
For every $\delta \in (0, 1)$, $c \geq 0$, and
sequence of i.i.d. random variables $Z_t$ with finite support and mean $\mu < 0$, and $Z_t$ takes positive value with positive probability, we have:
\begin{align*}
\Pr\left[\bigcup_{n=1}^\infty \left\{\sum_{t=1}^{n} \delta^t Z_t \geq c\right\}\right]
\leq \exp(-r^* \cdot c)
\end{align*}
where $r^* > 0$ is the smallest positive real number such that
$\expect[z\sim Z_1]{\exp(r^*z)} = 1$.
\end{Lemma}
Intuitively, Lemma \ref{lem:concentration}
bounds the probability with which the discounted sum of a sequence of random variables  deviates significantly from its expectation.
\begin{proof}[Proof of Lemma A.1:]
Let $\gamma_{Z, t}(r) = \ln \expect[z\sim Z_t]{\exp(rz \delta^t)}$, and let
\begin{align*}
q_{Z, r, t}(z) = p_{Z}(z)\exp(rz \delta^t - \gamma_{Z,t}(r)),
\end{align*}
where $p_{Z}(z)$ is the probability mass function of random variable $Z$.
One can verify that $q$ is a well-defined probability measure.
For a sequence of random variables $Z^n \equiv \{Z_1, \dots, Z_n\}$, we have
\begin{align*}
q_{Z^n, r}(z_1, \dots, z_n)
= p_{Z^n}(z_1, \dots, z_n) \exp
\left(\sum_{t=1}^n r z_t \delta^t -\sum_{t=1}^n\gamma_{Z_t,t}(r)\right).
\end{align*}
Let $s_n = \sum_{t=1}^n z_t\delta^t$,
we have
\begin{align*}
q_{S^n, r}(s_n)
= p_{S^n}(s_n) \exp
\left(r s_n -\sum_{t=1}^n\gamma_{Z_t,t}(r) \right).
\end{align*}
Since $q_{S^n, r}$ is a probability measure, we have
\begin{align}\label{eq:expectation}
\expect{\exp
\left(r s_n -\sum_{t=1}^n\gamma_{Z_t,t}(r) \right)}
= 1.
\end{align}
Let $\gamma(r) \equiv \expect[z\sim Z_1]{\exp(rz)}$,
we have $\gamma(0) = 1$ and $\gamma'(0) = \expect[z\sim Z_1]{z} < 0$.
Since $r^* > 0$ is the smallest positive real number such that
$\expect[z\sim Z_1]{\exp(r^*z)} = 1$,
we have $\gamma(r) \leq 1$ for any $0\leq r\leq r^*$.
Since random variables $Z_t$ are i.i.d.,
we have
$$\gamma_{Z_t,t}(r^*) = \ln \expect[z\sim Z_t]{\exp(r^*z\delta^t)}
= \ln \expect[z\sim Z_1]{\exp(r^*z\delta^t)} \leq 0$$
for every $t \geq 1$.
By substituting $r=r^*$ in inequality \eqref{eq:expectation},
we have
$\expect{\exp\left(r^* s_n\right)} \leq 1$.\footnote{Note that when $\delta = 1$,
the inequality holds with equality, which is the Wald's identity established in \citet{Wald-44}.}
Let $J$ be the stopping time that the sum $s_J$ first exceeds the threshold $c$,
we have
\begin{align*}
\Pr\left[s_J\geq  c\right]
\cdot \expect{\exp(r^* s_J) \Big| s_J \geq  c} \leq 1,
\end{align*}
which implies that
\begin{equation*}
\Pr\left[\bigcup_{n=1}^\infty \left\{\sum_{t=1}^{n} \delta^t Z_t \geq c\right\}\right]
= \Pr\left[s_J\geq  c\right]
\leq \exp(-r^*\cdot c). \qedhere
\end{equation*}
\end{proof}

Back to the illustration of our constructive proof. Let $Z$ be a random variable that equals $0$ with probability $\varepsilon_1 \alpha' (a^*)$, equals $u_1(a^*,b^*)-u_1(a',b^*)$ with probability $\varepsilon_1 \alpha' (a')$, and equals $u_1(a^*,b^*)-u_1(\alpha,b)$ with probability $1-\varepsilon_1$ where $(\alpha,b) \in \Gamma$ is drawn according to distribution $q$. Intuitively, $Z_t$ measures the difference between the stage-game payoff player $1$ receives in the absorbing subphase and his target payoff $u_1(a^*,b^*)$.

Since the support of $Z$ is a finite set and the expectation of $Z$ is negative,
we can
apply Lemma \ref{lem:concentration} to a sequence of random variables with distribution $Z$. Our lemma implies that once play enters the absorbing subphase,
the event that:
{\sloppy
\begin{itemize}
    \item $\sum_{t=0}^T \delta^t u_1(a_t,b_t)$ is between $(1-\delta^{T+1}) u_1(a^*,b^*)+c(1-\delta)$ and $(1-\delta^{T+1}) \Big(\varepsilon_1 u_1(\alpha',b^*) +(1-\varepsilon_1) \mathbb{E}_{(\alpha,b) \sim q} [u_1(\alpha,b)]\Big)-c(1-\delta)$
for all $T \in \mathbb{N}$,
\end{itemize}}
\noindent occurs with probability bounded away from $0$. Since all other phases end in finite time in expectation, the discounted frequency of player $1$'s action is close to his discounted action frequency in the absorbing subphase, which is at most $\varepsilon_1$ away from $\alpha^*$.

\paragraph{Remark on Public Randomization Device:} The public randomization device is introduced to ease the exposition. It can be dispensed in our constructive proof for a reason similar to that in \cite{FM-91}. In what follows, we provide an intuitive explanation based on the constructive proof of Theorem \ref{Theorem1} in Appendix \ref{sec5}. 
The details of the construction without public randomization is available upon request.

Recall $(\alpha_1,\alpha_2,b_1,b_2,q)$ which solves the constrained minimization problem that defines $F^*(u_1,u_2)$.
Intuitively, the public randomization device plays two roles.
First, it is used to implement particular interior action frequencies, i.e., those in which $a^*$ is played with frequency strictly between $0$ and $1$. 
For this purpose, it is sufficient to choose a sequence of pure actions under which the discounted frequency approximates that of the implemented mixed action. 
Second, the public randomization device  delivers the exact continuation payoff that makes player $1$ indifferent
by mixing between pure action profiles $(a,b) \in \Gamma$.
As shown in \citet{FM-91}, any payoff~$v$ can be decomposed as the discounted average payoff of 
an infinite sequence of deterministic pure action profiles $(a,b) \in \Gamma$ 
when player~$1$ is sufficiently patient. 
Therefore, our constructed  equilibrium can be sustained in absence of public randomization. 
Finally, for any $\varepsilon>0$, let $T$ be the time period such that $\delta^T= \epsilon$. 
When players have access to a public randomization device, 
we use the public randomization device by the end of each block to set the discounted average payoff exactly to $u_1(a^*,b^*)$. In environments without the public randomization device, we can immediately start the next block if that block ends before period $T$. 
For any block after period $T$, 
we replace the public randomization device with an infinite sequence of deterministic pure action profiles that exactly implements the desired discounted payoff.
Note that this does not affect the incentives of player $1$ for using mixed strategies 
because the payoff differences in earlier blocks will be rectified by the compensation phase in later blocks. 
The public randomization device can be dispensed 
since the discounted frequency of any action affected by replacing the public randomization device after time $T$ is at most $\epsilon$.

\section{Proofs of Statement 1 of Theorems 1 and  2}\label{secB}
We start from showing that Statement 1 of Theorem \ref{Theorem1} is implied by Statement 1 of Theorem \ref{Theorem2} by showing that it is without loss of generality to focus on $\{\alpha_1,\alpha_2,b_1,b_2,q\}$ such that (\ref{3.5}) is binding in the constrained optimization problem that defines $F^*(u_1,u_2)$. Let
\begin{equation}\label{B.1}
    F^{**} (u_1,u_2) \equiv \min_{(\alpha_1,\alpha_2,b_1,b_2,q) \in \Delta (A) \times \Delta (A) \times B \times B \times [0,1] }
  \Big\{   q \alpha_1(a^*) +(1-q) \alpha_2(a^*) \Big\},
\end{equation}
subject to
\begin{equation}\label{B.2}
(\alpha_1,b_1) \in \Gamma,\quad (\alpha_2,b_2) \in \Gamma,
\end{equation}
and
\begin{equation}\label{B.3}
    q u_1(\alpha_1,b_1) +(1-q) u_1(\alpha_2,b_2) = u_1(a^*,b^*),
\end{equation}
Compared to $F^*(u_1,u_2)$, the objective function and the first constraint remains the same, but the inequality constraint (\ref{3.5}) is replaced by the equality constraint (\ref{B.3}).
\begin{Lemma}\label{LB.1}
Suppose $(u_1,u_2)$ satisfies Assumptions 1 and 2, then $F^{**}(u_1,u_2)=F^*(u_1,u_2)$.
\end{Lemma}
\begin{proof}
The part in which $F^{**}(u_1,u_2) \geq F^*(u_1,u_2)$ is straightforward. Next, we show  $F^{**}(u_1,u_2) \leq F^*(u_1,u_2)$. 
Suppose the constrained minimum in (\ref{3.3}) is attained by $\{\alpha_1,\alpha_2,b_1,b_2,q\}$ where $q u_1(\alpha_1,b_1)+(1-q) u_1(\alpha_2,b_2) > u_1(a^*,b^*)$.
Since $a^*$ is player $1$'s unique Stackelberg action, for every $a' \neq a^*$, there exists $b' \in \textrm{BR}_2(a')$ such that $u_1(a',b')<u_1(a^*,b^*)$. Let $r \in [0,1]$ be defined via:
\begin{equation*}
    r u_1(a',b') +(1-r) \Big(u_1(\alpha_1,b_1)+(1-q) u_1(\alpha_2,b_2)\Big) = u_1(a^*,b^*).
\end{equation*}
Consider an alternative distribution $q' \in \Delta (\Gamma)$ that attaches probability $r$ to $(a',b')$, probability $(1-r)q$ to $(\alpha_1,b_1)$, and probability $rq$ to $(\alpha_2,b_2)$.
The probability of $a^*$ is weakly lower under $q'$ compared to that under $q$, and constraint (\ref{3.5}) is binding.  Later on, we show in Lemma \ref{LA.3}  that there exists a distribution over incentive compatible action profiles supported on two elements under which constraint (\ref{3.4}) is satisfied,  constraint (\ref{3.5}) is binding, and attains the constrained minimum. Therefore, $F^{**}(u_1,u_2) \leq F^*(u_1,u_2)$.
\end{proof}

In the remainder of this appendix, we show that the equilibrium constructed in Appendix  \ref{sec5}
achieves the desired (discounted) action frequencies.
We first define the parameters used in the construction of the equilibrium. 
Let $\varepsilon_1>0$ be a small positive real number,
and let $Z_1 = \commitu - u_1(a, b)$ be a random variable that
\begin{itemize}
    \item equals $\commitu - u_1(a^*,b^*)$ 
    with probability $\varepsilon_1 \alpha' (a^*)$,
    \item equals $\commitu- u_1(a',b^*)$ with probability $\varepsilon_1 \alpha'(a')$,
    \item with probability $1-\varepsilon_1$,
    equals $\commitu - u_1(a,b)$ where $(a,b)$ is distributed according to $q$.
\end{itemize}
One can verify that $Z_1$ has finite support and $\expect{Z_1} < 0$. 
Let $r^*_1 > 0$ be the smallest real number such that $\expect[z\sim Z_1]{\exp(r^*_1 \cdot z)} = 1$.\footnote{Here we consider the case that the random variable $Z_1$ takes positive value with positive probability. 
As will become clearer in the analysis, 
the case when $Z_1$ only has non-positive support is trivial. 
We made the same assumption for $Z_2$ as well.} 
Similarly, let $Z_2 = u_1(a, b) - \epsilon_1$ be the random variable that:
\begin{itemize}
    \item equals $u_1(a^*,b^*)-\epsilon_1$ with probability $\epsilon_1 \mixaction'(a^*)$,
    \item equals $u_1(a',b^*)-\epsilon_1$ with probability $\epsilon_1 \mixaction'(a')$,
    \item with probability $1-\epsilon_1$, equals $u_1(a,b)-\epsilon_1$ where $(a,b)$ is distributed according to $q$.
\end{itemize}
Let $r^*_2 > 0$ be the smallest real number such that $\expect[z\sim Z_2]{\exp(r^*_2 \cdot z)} = 1$.
Let $\overline{M} \equiv \max_{(a,b) \in A \times B} u_1(a,b)$
and let $T_1 = \lceil\frac{\overline{M}+c}{u_1(\deviatea,b^*) - \commitu}\rceil$
where $c \in \mathbb{R}_+$ is such that $\exp(-\min\{r^*_1, r^*_2\} \cdot c) \leq \epsilon_1$. 
Next we introduce several minor changes in the construction of the equilibrium in Appendix \ref{sec5} to simplify the exposition. 
\begin{itemize}
\item We impose a universal upper bound on the length of each absorbing subphase
as $\bar{T}_2 \equiv \lceil\frac{\ln (1-\epsilon_1)}{\ln \delta}\rceil$,
and let $T_2 \leq \bar{T}_2$ be the stopping time of the absorbing subphase.\footnote{$T_2$ is the number of period in the current absorbing subphase, not the time horizon. }

\item Letting $T_0$ be the starting time of the absorbing subphase, 
$T_2 \leq \bar{T}_2$ is the smallest interger that satisfies 
\begin{equation*}
        \sum_{t=0}^{T_2} \delta^t u_1(a_{t+T_0},b_{t+T_0}) < (1-\delta^{T_2+1}) u_1(a^*,b^*)-c(1-\delta)
    \end{equation*}
    or
    \begin{equation*}
      \sum_{t=0}^{T_2} \delta^t u_1(a_{t+T_0},b_{t+T_0}) 
       > (1-\delta^{T_2+1}) \Big(\varepsilon_1 u_1(\alpha',b^*) +(1-\varepsilon_1) \mathbb{E}_{(\alpha,b) \sim q} [u_1(\alpha,b)]\Big) + c(1-\delta).  
    \end{equation*}
\end{itemize}
The second bulletin point defines the stopping criterion based on the discounted average payoff within the absorbing subphase.
Moreover, the first inequality in the second bulletin is consistent with the constraint that the discounted average payoff from time $0$ to $T$ is above $u_1(a^*, b^*)$
because the accumulated payoff in the first $T_1$ periods of the current block is sufficiently high when we start the absorbing subphase.


Next we prove Statement 1 of Theorem 2 with the above parameters constructed in the equilibrium 
when $\delta >\overline{\delta}$ with
\begin{equation}
    \bar{\delta} = \max\left\{\frac{\ln(1-\epsilon_1^3)}{\ln T_1}, 1-\epsilon_1^2\right\}.
\end{equation}

In the equilibrium constructed in Appendix \ref{sec5}, 
the discounted payoff for player 1 in each block equals 
$(1-\delta^T)\commitu$, in which
$T \in \mathbb{N}$ is the number of time periods in the block. 
This implies that the strategic type has an incentive to play the mixed action in the beginning of the game to separate from the commitment type. 
In addition, one can verify that player $1$ has no incentive to make any off-path deviations, since his expected continuation value at every on-path history is strictly greater than
 $\minmaxu$ when
$\delta$ is sufficiently close to 1.

Let $\event_1$ be the event that 
player $1$'s discounted payoff in the absorbing subphase is less than 
 $(1-\delta^t)\commitu-c(1-\delta)$. 
Let $\event_2$ be the event that 
player $1$'s discounted payoff in the absorbing subphase is more than
$(1-\delta^t)(\epsilon_1 u_1(\mixaction',b^*) 
+ (1-\epsilon_1)\expect[(\mixaction,b)\sim q]{u_1(\mixaction,b)}
+\epsilon_1) +c(1-\delta)$.
First, the probability that event $\event_1$ happens 
is bounded from above
by the probability that 
$\sum_{t=1}^n \delta^t z_{1;t}$ is greater than $c$ for some $n \geq 1$
where $z_{1;t} \sim Z_1$ for all~$t$. 
According to Lemma \ref{lem:concentration}, 
the latter probability is bounded from above by $\exp(-r^*_1\cdot c) \leq \epsilon_1$, 
which implies that $\Pr[\event_1] \leq \epsilon_1$. 
Similarly, we have $\Pr[\event_2] \leq \epsilon_1$.
Let $\event_3$ be the event that action profile $(\deviatea, b^*)$ is observed for $T_1$ periods, 
and by definition we have $\Pr[\event_3] = \q^{T_1}$.

We first show that $G^{\strategy_1,\strategy_2}(a)
\leq \mixaction^*(a) + \epsilon$
for every $a \in A$. 
Let $G$ denote the discounted number of times action $a$ is chosen from the beginning of each block. 
By construction, we have 
\begin{align*}
G &\leq (1-\delta^{T_1}) 
+ (1-\q^{T_1}\cdot (1-2\epsilon_1))
\cdot \delta^{T_1} G
+ (1-2\epsilon_1)\cdot \q^{T_1} \delta^{T_1+\bar{T}_2} G
+ \q^{T_1} \delta^{T_1} (1-\delta^{\bar{T}_2}) 
(\epsilon_1 + (1-\epsilon_1)\mixaction^*(a))\\
\Rightarrow G &\leq 
\frac{1-\delta^{T_1} 
+ \q^{T_1} \delta^{T_1} (1-\delta^{\bar{T}_2})
(\epsilon_1 + (1-\epsilon_1)\mixaction^*(a))}
{(1-2\epsilon_1)(1-\delta^{\bar{T}_2})\delta^{T_1}\q^{T_1} + (1-\delta^{T_1})} 
\leq \frac{\mixaction^*(a) + \epsilon_1}{1-2\epsilon_1}.
\end{align*}
The first term in the first inequality is the upper bound on the discounted number of times action $a$ is chosen
from period $1$ to $T_1$; 
the second term is the upper bound on the discounted number of times action $a$ is chosen in future blocks 
conditional on event $(\event_1\cup\event_2)$ happens;
the third term is the upper bound on the discounted number of times action $a$ is chosen in future blocks 
conditional on event $\neg(\event_1\cup\event_2)$,
and the last term is the upper bound on the discounted number of times action $a$ is chosen in the absorbing subphase.
The second inequality holds by rearranging terms. 
By setting  $\epsilon_1 \ll \q^{T_1}$, 
the last inequality holds 
since $1-\delta^{T_1} \leq \epsilon_1^3$ and $1-\delta^{\bar{T}_2} \approx \epsilon_1$.
Therefore,
\begin{align*}
\mathbb{E}^{(\sigma_1,\sigma_2)} \Big[
    \sum_{t=0}^{\infty} (1-\delta)\delta^t \mathbf{1}\{a_{t}=a\}
    \Big]
&\leq \sum_{t=0}^\infty \q (1-\q)^t
\left(1-\delta^t + \delta^t G\right)\\
&= \frac{(1-\q)(1-\delta)}{1-(1-\q)\delta} 
+ \frac{\mixaction^*(a) + \epsilon_1}{(1-(1-\q)\delta)(1-2\epsilon_1)}
\leq \mixaction^*(a) + \epsilon. 
\end{align*}
where the last inequality holds for sufficiently small $0<\epsilon_1 \ll \epsilon$.

Next we show that $G^{\strategy_1,\strategy_2}(a)
\geq \mixaction^*(a) - \epsilon$
for every $a \in A$. 
First, we provide an upper bound for the stopping time $T$.
Conditional on event $\event_2 \cap \event_3$, 
the stopping time $T$ satisfies 
\begin{align}
(1-\delta^{T_1+T_2})&\overline{M}
+ \delta^{T_1+T_2} (1-\delta^{T-T_1-T_2}) u_1(a',b') 
\geq (1-\delta^{T}) \commitu&\nonumber\\
\Rightarrow
\delta^{T}
&\geq 
\frac{\commitu-\delta^{T_1+T_2}(u_1(\mixaction',b')
- (1-\delta^{T_1+T_2})\overline{M}}
{\commitu-u_1(\mixaction',b')}\nonumber\\
&\geq \delta^{T_1+T_2} -  \frac{(1-\delta^{T_1+T_2})\overline{M}}
{\commitu-u_1(\mixaction',b')}
\geq \delta^{T_1+\bar{T}_2} -  \frac{(1-\delta^{T_1+\bar{T}_2})\overline{M}}
{\commitu-u_1(\mixaction',b')}\label{eq:t1}
\end{align}
Conditional on event $(\neg\event_2) \cap \event_3$, 
the stopping time $T$ satisfies 
\begin{align}
(1-\delta^{T_1})&\overline{M}
+\delta^{T_1}(1-\delta^{T_2})(\epsilon_1 u_1(\mixaction',b^*)
+ (1-\epsilon_1)\expect[(\mixaction,b)\sim q]{u_1(\mixaction,b)}
+\epsilon_1) \nonumber\\
&+c(1-\delta)
+ \delta^{T_1+T_2} (1-\delta^{T-T_1-T_2}) u_1(a',b') 
\geq (1-\delta^{T}) \commitu&\nonumber\\
\Rightarrow
\delta^{T}
&\geq 
\frac{\delta^{T_1+T_2}(\commitu-(u_1(\mixaction',b'))
- (1-\delta^{T_1})\overline{M} - c(1-\delta) - \delta^{T_1}(1-\delta^{T_2})(\epsilon_1 u_1(\mixaction',b^*)+\epsilon_1)}
{\commitu-u_1(\mixaction',b')}\nonumber\\
&\geq \delta^{T_1+T_2} -  \frac{\epsilon^2_1(1-\delta^{T_1+T_2})(2\overline{M}+c)}
{\commitu-u_1(\mixaction',b')}
\geq \delta^{T_1+\bar{T}_2} -  \frac{\epsilon^2_1(1-\delta^{T_1+\bar{T}_2})(2\overline{M}+c)}
{\commitu-u_1(\mixaction',b')}\label{eq:t2}
\end{align}
Conditional on event $\neg \event_3$,
the stopping time $T$ satisfies 
\begin{align}
(1-\delta^{T_1})&\overline{M}
+ \delta^{T_1} (1-\delta^{T-T_1}) u_1(a',b') 
\geq (1-\delta^{T}) \commitu&\nonumber\\
\Rightarrow
\delta^{T}
&\geq 
\frac{\commitu-\delta^{T_1}u_1(\mixaction',b')
- (1-\delta^{T_1})\overline{M}}
{\commitu-u_1(\mixaction',b')}\nonumber\\
&\geq \delta^{T_1} -  \frac{(1-\delta^{T_1})\overline{M}}
{\commitu-u_1(\mixaction',b')}\label{eq:t3}
\end{align}


Let $G$ denote the discounted number of times action $a$ is chosen in each block. 
By construction, we have 
\begin{align*}
G &\geq (1-\q^{T_1})
(\delta^{T_1} -  \frac{(1-\delta^{T_1})\overline{M}}
{\commitu-u_1(\mixaction',b')}) G
+ \q^{T_1}(1-\epsilon_1) 
(\delta^{T_1+\bar{T}_2} -  \frac{\epsilon^2_1(1-\delta^{T_1+\bar{T}_2})(2\overline{M}+c)}
{\commitu-u_1(\mixaction',b')}) G \nonumber\\
&\quad + \q^{T_1} \epsilon_1 
(\delta^{T_1+\bar{T}_2} -  \frac{(1-\delta^{T_1+\bar{T}_2})\overline{M}}
{\commitu-u_1(\mixaction',b')}) G
+ \q^{T_1} \delta^{T_1} (1-\delta^{\bar{T}_2}) 
(1-\epsilon_1)\mixaction^*(a)\\
\Rightarrow G &\geq 
\frac{\q^{T_1} \delta^{T_1} (1-\delta^{\bar{T}_2})
(1-\epsilon_1)\mixaction^*(a)}
{\q^{T_1}\delta^{T_1}(1-\delta^{\bar{T}_2}) +  O(\epsilon^2_1)} 
\geq \frac{\mixaction^*(a) (1-\epsilon_1)}{1+O(\epsilon_1)}.
\end{align*}
The first term in the first inequality is the lower bound 
on the discounted number of times action $a$ is chosen in future blocks 
conditional on event $\neg\event_3$;
the second term is the lower bound on the discounted number of times action $a$ is chosen in future blocks 
conditional on event $\event_3\cap(\neg\event_2)$;
the third term is the lower bound on the discounted number of times action $a$ is chosen in future blocks 
conditional on event $\event_3\cap\event_2$;
and the last term is the lower bound on the discounted number of times action $a$ is chosen in absorbing subphase.
Finally, we have 
\begin{align*}
\mathbb{E}^{(\sigma_1,\sigma_2)} \Big[
    \sum_{t=0}^{\infty} (1-\delta)\delta^t \mathbf{1}\{a_{t}=a\}
    \Big]
&\geq \sum_{t=0}^\infty \q (1-\q)^t \delta^t G\\
&= \frac{\mixaction^*(a) (1-\epsilon_1)}
{(1-(1-\q)\delta)(1+O(\epsilon_1))}
\geq \mixaction^*(a) - \epsilon. 
\end{align*}
where the last inequality holds when $\epsilon_1$ is sufficiently small compared to $\epsilon$.
Combining these bounds, we have 
\begin{align*}
\abs{
\mathbb{E}^{(\sigma_1,\sigma_2)} \Big[
    \sum_{t=0}^{\infty} (1-\delta)\delta^t \mathbf{1}\{a_{t}=a\}
    \Big]
    - \mixaction^*(a)}
\leq \epsilon \quad \textrm{for every} \quad a \in A.
\end{align*}

\section{Proofs of Statement 2 of Theorems 1 and 2}\label{secA}
First, we establish Statement 2 of Theorem \ref{Theorem1}. Let $\Delta (\Gamma)$ be the set of probability distributions on $\Gamma$ 
whose support has countable number of elements. 
Let $F(u_1,u_2,\varepsilon)$ be the value of the following constrained optimization problem:
\begin{equation}\label{A.1}
F(u_1,u_2,\varepsilon) \equiv \inf_{p \in \Delta (\Gamma)}    \int \alpha (a^*) d p(\alpha,b),
\end{equation}
subject to
\begin{equation}\label{A.2}
    \int u_1(\alpha,b) d p(\alpha,b) \geq u_1(a^*,b^*)-\varepsilon.
\end{equation}
Our proof of the necessity part of Theorem \ref{Theorem1} consists of three lemmas.
\begin{Lemma}\label{LA.1}
For every $\pi>0$ and $\varepsilon>0$, there exists $\underline{\delta} \in (0,1)$ such that for every $\delta >\underline{\delta}$,
\begin{equation}\label{A.3}
 G^{(\sigma_1,\sigma_2)}(a^*) \geq F(u_1,u_2,\varepsilon)-(1-\underline{\delta}) \quad \textrm{for every} \quad (\sigma_1,\sigma_2) \in \textrm{NE}(\delta,\pi).
\end{equation}
\end{Lemma}
\begin{Lemma}\label{LA.2}
For every $u_1$ and $u_2$ that satisfy Assumptions \ref{Ass1} and \ref{Ass2}, $\lim_{\varepsilon \downarrow 0} F(u_1,u_2,\varepsilon)=F(u_1,u_2,0)$.
\end{Lemma}
\begin{Lemma}\label{LA.3}
For every $u_1$ and $u_2$ that satisfy Assumptions \ref{Ass1} and \ref{Ass2}, $F^*(u_1,u_2)=F(u_1,u_2,0)$.
\end{Lemma}

\begin{proof}[Proof of Lemma C.1:] The reputation result in \citet{FL-89} implies that for every $\pi>0$ and $\varepsilon>0$, there exists $\underline{\delta} \in (0,1)$ such that for every $\delta >\underline{\delta}$,
\begin{equation}\label{A.4}
 \mathbb{E}^{(\sigma_1,\sigma_2)}\Big[
    \sum_{t=0}^{+\infty} (1-\delta) \delta^t u_1(a_t,b_t)
    \Big] \geq u_1(a^*,b^*)- \varepsilon/2 \textrm{ for every } (\sigma_1,\sigma_2) \in \textrm{NE}(\delta,\pi).
\end{equation}
For given $(\sigma_1,\sigma_2) \in \textrm{NE}(\delta,\pi)$, let $\mathcal{H}^*$ be a set of on-path histories such that $h^t \in \mathcal{H}^*$ if and only if
\begin{itemize}
  \item $a^*$ was played from period $0$ to $t-1$,
\textit{and} $\sigma_1(h^t)$ assigns positive probability to actions other than $a^*$.
\end{itemize}
By construction, for every $h^t \in \mathcal{H}^*$,  player $2$'s posterior belief at $h^t$ assigns probability at least $\pi$ to the commitment type, and therefore, player $1$'s continuation value at $h^t$ is at least $u_1(a^*,b^*)-\varepsilon/2$. Let $\overline{M} \equiv \max_{(a,b) \in A \times B} u_1(a,b)$.
For every $a \in \textrm{supp}(\sigma_1(h^t)) \backslash \{a^*\}$ and $b \in \textrm{supp}(\sigma_2(h^t))$, player $1$'s continuation value at $(h^t,a,b)$, denoted by $v(h^t,a,b)$, satisfies:
\begin{equation*}
v(h^t , a,b) \geq \frac{1}{\delta} \Big(
u_1(a^*,b^*)-\frac{\varepsilon}{2}-(1-\delta)\overline{M}
\Big).
\end{equation*}
The right-hand-side is strictly greater than $u_1(a^*,b^*)-\varepsilon$ when $\delta$ is close enough to $1$. For every on-path history $h^s$ such that $h^s \succeq (h^t,a,b)$, player $2$ attaches probability $1$ to the rational type at $h^s$, and therefore, $\sigma_2(h^s)$ best replies against $\sigma_1(h^s)$. Therefore, 
$(\sigma_1(h^s),b) \in \Gamma$ for every $b \in \textrm{supp}(\sigma_2(h^s))$. Let $p_{(h^t,a,b)} \in \Delta (\Gamma)$ 
be a probability measure on $\Gamma$ such that for every $(\alpha,b) \in \Gamma$, 
  \begin{equation}\label{A.5}
    p_{(h^t,a,b)} (\alpha,b) \equiv \mathbb{E}^{(\sigma_1,\sigma_2)} \Big[
    \sum_{s=t+1}^{\infty} (1-\delta) \delta^{s-t-1} \mathbf{1}\{\sigma_1(h^s)=\alpha \} \sigma_2(b)\Big| (h^t,a,b)
    \Big].
  \end{equation}
By construction, $p_{(h^t,a,b)}$ has a countable number of elements in its support, and 
player $1$'s continuation value at $(h^t,a,b)$, denoted by $v(h^t,a,b)$, satisfies
\begin{equation}\label{A.6}
   v(h^t , a,b)= \int u_1(\alpha,b) d p_{(h^t,a,b)} (\alpha,b) \geq u_1(a^*,b^*)-\varepsilon.
\end{equation}
The definition of $F(u_1,u_2,\varepsilon)$ in (\ref{A.1}) and (\ref{A.2}) suggests that:
\begin{equation}\label{A.7}
G^{(h^t,a,b)}(a^*) \equiv    \mathbb{E}^{(\sigma_1,\sigma_2)} \Big[
    \sum_{s=t+1}^{\infty} (1-\delta) \delta^{s-t-1} \mathbf{1}\{a_s=a^*\} \Big| (h^t,a,b) \Big] \geq F(u_1,u_2,\varepsilon).
\end{equation}
Next, we compute a lower bound on $G^{(\sigma_1,\sigma_2)}(a^*)$. Let $\widehat{\mathcal{H}}$ be the set of on-path histories $h^t \equiv (h^{t-1},a_{t-1},b_{t-1})$ such that $t \geq 1$, $h^{t-1} \in \mathcal{H}^*$, and $a_{t-1} \neq a^*$. Let $p^{(\sigma_1,\sigma_2)} (h^t)$ be the ex ante probability of history $h^t$ under the probability measure induced by $(\sigma_1,\sigma_2)$. By definition,
    $1-\sum_{h^t \in \widehat{\mathcal{H}}} p^{(\sigma_1,\sigma_2)} (h^t)$
is the ex ante probability with which player $1$ plays $a^*$ in every period conditional on him being the rational type. Therefore, 
\begin{eqnarray}\label{A.8}
G^{(\sigma_1,\sigma_2)} (a^*) & = & \Big(1-\sum_{h^t \in \widehat{\mathcal{H}}} p^{(\sigma_1,\sigma_2)} (h^t)\Big)
+\sum_{h^t \in \widehat{\mathcal{H}}} p^{(\sigma_1,\sigma_2)} (h^t) \Big((1-\delta^{t-1})+\delta^t X^{(h^t)} (a^*)\Big)
{}
\nonumber\\
&\geq & {} -(1-\delta)+\Big(1-\sum_{h^t \in \widehat{\mathcal{H}}} p^{(\sigma_1,\sigma_2)} (h^t)\Big)
+\sum_{h^t \in \widehat{\mathcal{H}}} p^{(\sigma_1,\sigma_2)} (h^t) \Big((1-\delta^{t})+\delta^t X^{(h^t)} (a^*)\Big)
{}
\nonumber\\
&\geq & {}  F(u_1,u_2,\varepsilon) -(1-\delta) \geq F(u_1,u_2,\varepsilon)-(1-\underline{\delta})
\end{eqnarray}
\end{proof}
\begin{proof}[Proof of Lemma C.2:] By definition, the value of $F(u_1,u_2,\varepsilon)$ is a decreasing function of $\varepsilon$ and is bounded by $[0,1]$. Therefore,
$\lim_{\varepsilon \downarrow 0} F(u_1,u_2,\varepsilon)$ exists and moreover, 
$\lim_{\varepsilon \downarrow 0} F(u_1,u_2,\varepsilon) \leq F(u_1,u_2,0)$.

Next, we show that 
$\lim_{\varepsilon \downarrow 0} F(u_1,u_2,\varepsilon) \geq F(u_1,u_2,0)$. 
The optimization problem that defines $F(u_1,u_2,\varepsilon)$ implies that
for every $\varepsilon>0$, there exists $p_{\varepsilon} \in \Delta (\Gamma)$ that has countable number of elements in its support such that $\int \alpha (a^*) d p_{\varepsilon}(\alpha,b) \leq F(u_1,u_2,\varepsilon)+ \varepsilon$ and $\int u_1(\alpha,b) d p_{\varepsilon}(\alpha,b) \geq u_1(a^*,b^*)-\varepsilon$. 

According to Assumption \ref{Ass2}, there exists $a' \in A$ such that $u_1(a',b^*) > u_1(a^*,b^*)$. According to Assumption \ref{Ass1}, $b^*$ is player $2$'s strict best reply against $a^*$. This implies the existence of $\alpha^* \in \Delta (A)$ such that $\alpha^* (a^*) \neq 1$, $b^* \in \textrm{BR}_2(\alpha^*)$, and $u_1(\alpha^*,b^*) > u_1(a^*,b^*)$.  Let $\rho \equiv u_1(\alpha^*,b^*) - u_1(a^*,b^*)$. Since the support of $p_{\varepsilon}$ is countable, there exists $\alpha_{\varepsilon}^* \in \Delta (A)$ such that 
$\alpha_{\varepsilon}^* (a^*) \neq 1$, $b^* \in \textrm{BR}_2(\alpha_{\varepsilon}^*)$,  $u_1(\alpha_{\varepsilon}^*,b^*) - u_1(a^*,b^*) \geq \frac{\rho}{2}$, and $(\alpha_{\varepsilon}^*, b^*)$ does not belong to the support of $p_{\varepsilon}$. We construct probability measure  $p_{\varepsilon}' \in \Delta (\Gamma)$ according to:
\begin{itemize}
\item $p_{\varepsilon}'(\alpha_{\varepsilon}^*,b^*)\equiv \frac{2\varepsilon}{\rho+2\varepsilon}$. 
  \item $p_{\varepsilon}' (\alpha,b) \equiv \frac{\rho}{\rho+2 \varepsilon} p_{\varepsilon} (\alpha,b)$ for every $(\alpha,b)$ that belongs to the support of $p_{\varepsilon}$.
\end{itemize}
By construction, $\int u_1(\alpha,b) d p_{\varepsilon}'(\alpha,b) \geq u_1(a^*,b^*)$, and therefore,
\begin{equation}\label{A.9}
  \frac{2\varepsilon}{ \rho+2\varepsilon} +  \frac{\rho}{ \rho+2\varepsilon} \Big( F(u_1,u_2,\varepsilon)+ \varepsilon \Big)
   \geq \int \alpha (a^*) d p_{\varepsilon}'(\alpha,b) \geq F(u_1,u_2,0).
\end{equation}
This implies that
\begin{equation*}
  \lim_{\varepsilon \downarrow 0}   \Big\{  \frac{2\varepsilon}{ \rho+2\varepsilon} +  \frac{\rho}{ \rho+2\varepsilon} \Big( F(u_1,u_2,\varepsilon)+ \varepsilon \Big) \Big\}
  =\lim_{\varepsilon \downarrow 0}  F(u_1,u_2,\varepsilon) \geq F(u_1,u_2,0).
\end{equation*}
\end{proof}
\begin{proof}[Proof of Lemma C.3:] The inequality that $F^*(u_1,u_2) \geq F(u_1,u_2,0)$ is implied by the definitions of 
$F^*(u_1,u_2)$ and $F(u_1,u_2,0)$. In what follows, we show that
$F^*(u_1,u_2) \leq F(u_1,u_2,0)$. For every $\eta >0$, there exists $p_{\eta} \in \Delta (\Gamma)$ that has countable number of elements in its support such that $ \int \alpha (a^*) d p_{\eta}(\alpha,b) \leq F(u_1,u_2,0)+\eta$ and $\int u_1(\alpha,b) d p_{\eta}(\alpha,b) \geq u_1(a^*,b^*)$. 
Let $\Gamma_{\eta}$ be a countable subset of $\Gamma$ that contains the support of $p_{\eta}$. Consider the following minimization problem:
\begin{equation}\label{A.10}
F_{\eta} \equiv \min_{p \in \Delta (\Gamma_{\eta})}   \sum_{(\alpha,b) \in \Gamma_{\eta}} p (\alpha,b) \alpha (a^*),
\end{equation}
subject to
\begin{equation}\label{A.11}
    \sum_{(\alpha,b) \in \Gamma_{\eta}} p (\alpha,b) u_1(\alpha,b)  \geq u_1(a^*,b^*).
\end{equation}
By construction, $F_{\eta} \leq \int \alpha (a^*) d p_{\eta}(\alpha,b) \leq F(u_1,u_2,0)+\eta$. We show that $F_{\eta}$ can be attained via a distribution that contains at most two elements in its support. The Lagrangian of the minimization problem is:
\begin{equation}\label{A.12}
    \sum_{(\alpha,b) \in \Gamma_{\eta}} p (\alpha,b) \alpha (a^*) +\lambda \Big( \sum_{(\alpha,b) \in \Gamma_{\eta}} p (\alpha,b) u_1(\alpha,b) - u_1(a^*,b^*)\Big),
\end{equation}
where $\lambda$ is the Lagrange multiplier. If constraint (\ref{A.11}) is not binding, then the minimum is zero and is attained by a degenerate distribution. If constraint (\ref{A.11}) is binding, then for every pair of elements $(\alpha,b)$ and $(\alpha',b')$ in the support of the minimand $p_{\eta}^* \in \Delta (\Gamma_{\eta})$, 
\begin{equation}\label{A.13}
    \alpha (a^*) +\lambda u_1(\alpha,b) =  \alpha' (a^*) +\lambda u_1(\alpha',b).
\end{equation}
Label the elements in the support of $p_{\eta}^*$ as $\{(\alpha_i,b_i)\}_{i=1}^{+\infty}$. Equation (\ref{A.13}) implies that for every $\alpha_i (a^*)\neq \alpha_j(a^*)$,
\begin{equation}\label{A.14}
    \frac{u_1(\alpha_i,b)-u_1(\alpha_j,b)}{\alpha_i(a^*)-\alpha_j(a^*)}=-\frac{1}{\lambda}.
\end{equation}
Let 
\begin{equation*}
    \overline{u}_1 \equiv \sup_{(\alpha,b) \in \{(\alpha_i,b_i)\}_{i=1}^{+\infty}} u_1(\alpha,b), \quad \underline{u}_1 \equiv \inf_{(\alpha,b) \in \{(\alpha_i,b_i)\}_{i=1}^{+\infty}} u_1(\alpha,b),
\end{equation*}
\begin{equation*}
    \overline{q} \equiv \sup_{(\alpha,b) \in \{(\alpha_i,b_i)\}_{i=1}^{+\infty}} \alpha (a^*), \quad \textrm{and} \quad \underline{q} \equiv \inf_{(\alpha,b) \in \{(\alpha_i,b_i)\}_{i=1}^{+\infty}} \alpha (a^*).
\end{equation*}
Equation (\ref{A.14}) implies that
\begin{equation*}
    \frac{\overline{u}-\underline{u}}{\overline{q}-\underline{q}}=-\frac{1}{\lambda}.
\end{equation*}
Let $\gamma \in (0,1)$ be such that $\gamma \overline{u}_1 +(1-\gamma)\underline{u}_1=u_1(a^*,b^*)$. According to (\ref{A.14}), we have $\gamma \overline{q} +(1-\gamma) \underline{q}= F_{\eta}$.

Since $\Delta (A) \times B$ is compact, there exist $(\overline{\alpha},\overline{b})$ and
$(\underline{\alpha},\underline{b})$ which are limit points of set $\{(\alpha_i,b_i)\}_{i=1}^{+\infty}$
such that $u_1(\overline{\alpha},\overline{b})=\overline{u}_1$, $\overline{\alpha}(a^*)=\overline{q}$,
$u_1(\underline{\alpha},\underline{b})=\underline{u}_1$, and $\underline{\alpha}(a^*)=\underline{q}$. 
Since player $2$'s best reply correspondence is upper-hemi-continuous, $(\overline{\alpha},\overline{b}), (\underline{\alpha},\underline{b}) \in \Gamma$. Our analysis above suggests that there exists a distribution on $\Gamma_{\eta} \bigcup \{(\overline{\alpha},\overline{b}) , (\underline{\alpha},\underline{b})\}$ with at most two elements in its support that satisfies constraint (\ref{A.2}) and the value of the objective function (\ref{A.1}) is at most $F(u_1,u_2,0)+\eta$.

Take a decreasing sequence of positive real numbers $\{\eta_n \}_{n \in \mathbb{N}}$ such that $\lim_{n \rightarrow \infty} \eta_n=0$. For every $n \in \mathbb{N}$, there exists $p_n \in \Delta (\Gamma)$ with at most two elements in its support that satisfies constraint (\ref{A.2}) and the value of the objective function is at most $F(u_1,u_2,0)+\eta_n$. Since $\Big(\Delta (A_1) \times B \Big)^2$ is compact, there exists a converging subsequence $\{p_{k_n}\}_{n \in \mathbb{N}}$ such that its limit $p^*$ has at most two elements in its support, satisfies constraint (\ref{A.2}), and the value of the objective function is at most $F(u_1,u_2,0)$. This implies that
$F^*(u_1,u_2) \leq F(u_1,u_2,0)$.
\end{proof}
In the last step, we modify the above proof in order to establish Statement 2 of Theorem \ref{Theorem2}. 
Since $\overline{v}_1 = u_1(a^*,b^*)$, for every $\varepsilon>0$, there exists $\underline{\delta} \in (0,1)$ such that player $1$'s payoff in every equilibrium where $\delta > \underline{\delta}$ is no more than $u_1(a^*,b^*)+\varepsilon$.  
Let
\begin{equation}
    \mathcal{A}^{\varepsilon} \equiv \Big\{
    \alpha^* \in \Delta(A)
    \Big|
    \exists q \in \Delta (\Gamma) \textrm{ such that }
\alpha^*    = \int_{\alpha} \alpha d q
\textrm{ and } \Big| \int_{(\alpha,b)} u_1(\alpha,b) d q 
- u_1(a^*,b^*) \Big| \leq \varepsilon
    \Big\}.
\end{equation}
Lemma \ref{LA.1} implies that for every $\alpha' \notin \mathcal{A}^{\varepsilon}$, there exist $\eta>0$ and $\underline{\delta} \in (0,1)$ such that for every $\delta > \underline{\delta}$ and every
$(\sigma_1,\sigma_2) \in \textrm{NE}(\delta,\pi)$, we have:
\begin{equation}
  \Big|  G^{(\sigma_1,\sigma_2)}(a)- \alpha^* (a) \Big| > \eta \textrm{ for some } a \in A.
\end{equation}
The conclusion of Theorem \ref{Theorem2} is obtained since $\lim_{\varepsilon \rightarrow 0} \mathcal{A}^{\varepsilon}=\mathcal{A}$.

\section{Proof of Proposition 1}\label{secC}
First, suppose toward a contradiction that
$\{\alpha_1,\alpha_2,b_1,b_2,q\}$ solves (\ref{3.3}),
player $2$ has a strict incentive to play $b_1$ against $\alpha_1$, $\alpha_1$ does not attach probability $1$ to player $1$'s lowest action, and $q \neq 0$.  One can increase the probability of $\underline{a}$ in $\alpha_1$ and
 decrease the probability of other actions, after which player $1$'s expected payoff strictly increases and the probability of action $a^*$ strictly decreases. This contradicts the presumption that $\{\alpha_1,\alpha_2,b_1,b_2,q\}$ solves the constrained minimization problem. 
 
 Next, suppose $q>0$ and $\alpha_1$ is such that $|\textrm{BR}_2(\alpha_1)| \geq 2$, and there exists $b \in \textrm{BR}(\alpha_1)$ such that $u_1(\alpha_1,b)> u_1(\alpha_1,b_1)$. Then replace $b_1$ by $b$ in the constrained minimization problem, the value of $F^*$ remains unchanged but constraint (\ref{3.5}) becomes slack. This contradicts Lemma \ref{LB.1} that it is without loss of generality to focus on $\{\alpha_1,\alpha_2,b_1,b_2,q\}$  such that constraint (\ref{3.5}) binds.

\end{spacing}

\bibliographystyle{plainnat}
\bibliography{ref}
\end{document}